\documentclass[journal,onecolumn,12pt]{IEEEtran}
\usepackage[letterpaper,margin=1in]{geometry}

\usepackage{amsmath,amssymb,amsthm,mathtools}
\usepackage{bm}
\usepackage{cite}
\usepackage[colorlinks=true,citecolor=blue,urlcolor=blue,linkcolor=blue]{hyperref}

\allowdisplaybreaks[4]
\numberwithin{equation}{section}

\newtheorem{theorem}{Theorem}[section]
\newtheorem{lemma}{Lemma}[section]
\newtheorem{proposition}{Proposition}[section]
\newtheorem{corollary}{Corollary}[section]
\newtheorem{remark}{Remark}[section]
\newtheorem{example}{Example}[section]

\renewcommand{\P}{\mathbb P}
\newcommand{\E}{\mathbb E}
\newcommand{\R}{\mathbb R}
\newcommand{\Cov}{\operatorname{Cov}}
\newcommand{\tr}{\operatorname{tr}}
\newcommand{\Law}{\mathcal L}
\newcommand{\HS}{\mathrm{HS}}
\newcommand{\op}{\mathrm{op}}
\newcommand{\Id}{I_d}
\newcommand{\etaVec}{\boldsymbol\eta}
\newcommand{\Jmat}{\mathbf J}
\newcommand{\Pp}{\mathbb P}
\newcommand{\indep}{\mathrel{\perp\!\!\!\perp}}

\begin{document}

\title{Conditional Fisher-Information Central Limit Theorems under Log-Concavity with Information-Theoretic Consequences}

\author{Tong Ye and Liu-Quan Yao%
}

\maketitle

\begin{abstract}
We establish conditional central limit theorems in Fisher information under
log-concavity in every fixed dimension. For conditionally centered normalized
sums, after whitening by the averaged conditional covariance, the averaged
conditional Fisher information converges to the dimension if and only if it is
finite at one convolution level. The scalar criterion follows as the
one-dimensional case; we also provide an independent scalar proof based on a
second-order continuity theorem for Fisher production on Gaussian-smoothed,
tail-controlled classes. For the original sums, the averaged Fisher
information matrix converges in operator norm to the inverse averaged
conditional covariance. Consequently, the conditional relative Fisher
information with respect to the limiting Gaussian law vanishes, and the
Gaussian logarithmic Sobolev inequality yields convergence in conditional
relative entropy and conditional entropy. We give two operational
consequences. For any fixed finite-constellation low-power input, the
first-order conditional mutual-information slope converges to the
Gaussian-noise benchmark. For Gaussian signaling at any fixed signal
covariance, the mutual-information gap from that benchmark is bounded by the
conditional relative Fisher deficit and hence vanishes asymptotically.
\end{abstract}


\section{Introduction}
\label{sec:introduction}

Convergence in Fisher information is one of the strongest
information-theoretic forms of the central limit theorem. For a random
variable $X$, write $I(X)$ for its Fisher information; the precise weak
Sobolev definition used for possibly nonsmooth densities is given in
Section~\ref{sec:information-functionals}. If $X_1,X_2,\ldots$ are independent
and identically distributed real-valued random variables with variance
$\sigma^2$ and
\[
   W_n=\frac1{\sqrt n}\sum_{i=1}^n X_i,
\]
the Fisher-information central limit problem asks for conditions under which
\[
   I(W_n)\longrightarrow\frac1{\sigma^2}.
\]
This convergence is substantially stronger than weak convergence and is
closely connected, through de Bruijn's identity, to entropy dissipation along
the heat flow. Fisher-information inequalities and their role in the central
limit theorem were developed systematically by Johnson and Barron
\cite{johnson2004fisher}, building on the classical information inequalities
of Stam and Blachman \cite{stam1959some,Blachman1965}. In the scalar
i.i.d.\ setting, Johnson and Barron also showed that if the Fisher
information of a normalized sum is finite at some convolution level, then it
converges to the Gaussian value; see also Bobkov, Chistyakov, and G\"otze
\cite{bobkov2014fisher} for refined asymptotic results and convergence rates
under moment assumptions.

We study the corresponding conditional problem when each summand is observed
together with a random state. Let $(\xi,\eta)$ be a random pair, let
$\{(\xi_i,\eta_i)\}_{i\ge1}$ be independent copies, and define
\begin{equation}
   Z_i=\xi_i-\E[\xi_i\mid\eta_i],
   \qquad
   S_n=\frac1{\sqrt n}\sum_{i=1}^n Z_i,
   \qquad
   \etaVec_n=(\eta_1,\ldots,\eta_n).
   \label{eq:intro-normalized-sum}
\end{equation}
Conditionally on $\etaVec_n$, the summands are independent but generally not
identically distributed, because their conditional variances and shapes
depend on the realized state. The central question is whether the averaged
conditional Fisher information of $S_n$ converges to that of the Gaussian
having the averaged conditional covariance.

The mode of Gaussian convergence sought here is stronger than that
provided by the entropic conditional central limit theorem.
Ma et al.~\cite{ma2024entropic} proved an entropic conditional CLT under
a finite expected conditional Fisher-information assumption.
Ye et al.~\cite{ye2026finiteentropy} subsequently identified a
finite-entropy criterion and established continuity of Fisher information
on Gaussian-smoothed, tail-controlled classes. As a technical step in the
latter work, they also obtained convergence of averaged conditional Fisher
information for sequences with a fixed independent Gaussian component,
under suitable tail and asymptotic assumptions. 
In contrast, the present work establishes Fisher-information
 convergence for the original conditionally log-concave sums, 
 without imposing an a priori Gaussian regularization, and
  characterizes this convergence by finiteness at a single convolution level. 
  The result holds in every fixed dimension and further yields operator-norm 
  convergence of the averaged Fisher information matrix. Relative-entropy 
  convergence alone, however, does not in general imply convergence of Fisher information, 
  so the Fisher-information CLT obtained here gives a stronger mode of Gaussianization.

Let $\xi\in\R^d$, where $d\ge1$ is fixed, and put
\[
   \Sigma=\E\Cov(\xi\mid\eta)\succ0,
   \qquad
   W_n=\Sigma^{-1/2}S_n.
\]
For a random vector $X$, write $\Jmat(X)$ for its Fisher information matrix and
$I(X)=\tr\Jmat(X)$. When the score exists,
\[
   \Jmat(X)
   =\E\!\left[\rho_X(X)\rho_X(X)^{\mathsf T}\right].
\]
Precise weak definitions and conditional conventions are given in
Sections~\ref{sec:information-functionals} and
\ref{sec:conditional-information-functionals}. 
Under full-dimensional conditional log-concavity, we prove the equivalence
\[
   \E I(W_n\mid\etaVec_n)\longrightarrow d
   \qquad\Longleftrightarrow\qquad
   \E I(W_{n_0}\mid\etaVec_{n_0})<\infty
   \quad\text{for some }n_0\ge1.
\]
Moreover,
\begin{equation}
   \E\Jmat(S_n\mid\etaVec_n)
   \longrightarrow\Sigma^{-1}
   \label{eq:intro-matrix-limit}
\end{equation}
in operator norm. The matrix statement retains directional information that
is lost after taking the trace. When $d=1$, this theorem immediately gives
\[
   \E I(S_n\mid\etaVec_n)\longrightarrow\frac1{\sigma^2},
   \qquad
   \sigma^2=\E\operatorname{Var}(\xi\mid\eta),
\]
under the same finite-at-one-level criterion.

Fisher-information matrices have also appeared in earlier multivariate
limit theory. Mayer-Wolf~\cite{mayerwolf1990cramer} studied continuity and
convexity properties of the Fisher matrix and used the Cram\'er--Rao
functional as a variational tool for convergence to Gaussian laws.
More recently, Eskenazis and Gavalakis
\cite{eskenazis2024gaussian} obtained quantitative operator-norm bounds
for the normalized Fisher information matrix of weighted sums of i.i.d.\
multivariate Gaussian mixtures. The two settings are complementary. 
Their result is quantitative and tracks
dimension dependence under a Gaussian-mixture representation. The present
result is qualitative for each fixed dimension, but it treats a conditional
random environment, in which the summands are independent but generally
non-identically distributed after conditioning, and derives the matrix limit
from a finite-convolution-level criterion.

Let $\gamma_\Sigma=N(0,\Sigma)$. We write
$\mathcal D_\Sigma(X\mid Y)$ and $\mathcal J_\Sigma(X\mid Y)$ for the
averaged conditional relative entropy and the dimensionless averaged
conditional relative Fisher information, respectively, with respect to the
fixed Gaussian reference $\gamma_\Sigma$; precise definitions are given in
Section~\ref{sec:relative-definitions}. For the centered sums considered here,
expanding the Gaussian score gives the exact identity
\begin{equation}
   \mathcal J_\Sigma(S_n\mid\etaVec_n)
   =
   \tr\!\left(\Sigma\,\E\Jmat(S_n\mid\etaVec_n)\right)-d.
   \label{eq:intro-relative-fisher-identity}
\end{equation}
Consequently, \eqref{eq:intro-matrix-limit} implies
$\mathcal J_\Sigma(S_n\mid\etaVec_n)\to0$. The Gaussian logarithmic Sobolev
inequality then gives
\[
   \mathcal D_\Sigma(S_n\mid\etaVec_n)
   \le\frac12\mathcal J_\Sigma(S_n\mid\etaVec_n),
\]
so the Fisher-information CLT also yields the corresponding conditional
relative-entropy and entropy convergence.

The matrix limit also has a direct operational interpretation for
state-dependent additive-noise channels. Consider a
fixed finite-constellation low-power input $U\in\R^d$, independent of the
aggregate noise and the state, with $\Cov(U)=Q$. For
\[
   Y_{n,\rho}=\sqrt\rho\,U+S_n,
   \qquad
   R_{n,U}(\rho)=I(U;Y_{n,\rho}\mid\etaVec_n),
\]
the conditional weak-input expansion gives
\begin{equation}
   R_{n,U}'(0+)
   =\frac12\tr\!\left(Q\,\E\Jmat(S_n\mid\etaVec_n)\right).
   \label{eq:intro-low-snr-slope}
\end{equation}
Thus the first-order slope depends on the finite constellation only through
its covariance. Combining \eqref{eq:intro-low-snr-slope} with
\eqref{eq:intro-matrix-limit} yields
\[
   R_{n,U}'(0+)
   \longrightarrow
   \frac12\tr(Q\Sigma^{-1}),
\]
the Gaussian benchmark for the first-order slope. We also derive a
complementary result for Gaussian signaling at a fixed signal covariance: the
gap to the Gaussian-noise benchmark is bounded by one half of the relative
Fisher deficit.

The fixed-dimensional proof uses a conditional Lindeberg--Feller theorem to
identify the conditional Gaussian limit, followed by Gaussian smoothing and a
second-order Fisher-dissipation argument. Since this theorem includes $d=1$,
the scalar criterion is stated as a corollary in the main text. For comparison,
Appendix~\ref{app:scalar-alternative} gives an independent scalar proof along
the entropic route of Ma et al.~\cite{ma2024entropic} and the
Gaussian-smoothed continuity method of \cite{ye2026finiteentropy}. The new
ingredient in that alternative argument is continuity of the second-order
Fisher production $K$, proved in Appendix~\ref{app:smoothed-continuity}.

The rest of the paper is organized as follows.
Section~\ref{sec:preliminaries} introduces the information functionals and
conditional conventions. Section~\ref{sec:main-results} states the unified
fixed-dimensional criterion, its scalar and matrix consequences, and the
relative-information consequences. Section~\ref{sec:multivariate-proof}
proves the criterion. Section~\ref{sec:applications} develops the weak-signal
and Gaussian-signaling applications. The technical second-order Fisher,
alternative scalar, conditional Lindeberg, and Gaussian-smoothing arguments
are collected in the appendices.

\section{Preliminaries and notation}
\label{sec:preliminaries}

Throughout the paper, all random variables and random vectors are defined on
a common probability space. Whenever a statement is formulated on $\R^d$,
the dimension $d\ge1$ is fixed. Whenever conditional densities are used, we
fix a jointly measurable version of the relevant conditional density.

\subsection{Basic notation and moments}

Let $f$ be a probability density on $\R^d$ with finite second moment. We write
\begin{equation}
   m_f=\int_{\R^d}x f(x)\,dx,
   \qquad
   \Cov(f)=\int_{\R^d}(x-m_f)(x-m_f)^{\mathsf T}f(x)\,dx.
\end{equation}
We also define the centered second-moment tail functional
\begin{equation}
   \tau_f(R)
   =\int_{\R^d}\|x-m_f\|^2
   \mathbf1_{\{\|x-m_f\|\ge R\}}f(x)\,dx,
   \qquad R\ge0.
   \label{eq:tail-functional}
\end{equation}
For a random vector $X$ with density $f_X$, we use the corresponding notation
$m(X)$, $\Cov(X)$, and $\tau_X(R)$. When $d=1$, $\Cov(X)$ is written as
$\operatorname{Var}(X)$.

\subsection{Information functionals}
\label{sec:information-functionals}

The differential entropy of a probability density $f$ on $\R^d$ is
\begin{equation}
   h(f)=-\int_{\R^d}f(x)\log f(x)\,dx,
\end{equation}
whenever the integral is well defined. For a random vector $X$ with density
$f_X$, we write $h(X)=h(f_X)$.

To allow nonsmooth and compactly supported densities, the scalar-valued
Fisher information is understood in the weak Sobolev sense
\begin{equation}
I(f)
\triangleq
\begin{cases}
4\displaystyle\int_{\R^d}\|\nabla\sqrt f(x)\|^2\,dx,
&
\sqrt f\in H^1(\R^d),
\\[1.2ex]
+\infty,
&
\text{otherwise}.
\end{cases}
\label{eq:weak-fisher}
\end{equation}
Here
\[
H^1(\R^d)
=
\left\{
u\in L^2(\R^d):
\partial_j u\in L^2(\R^d),\
j=1,\ldots,d
\right\},
\]
where the derivatives are understood in the weak sense.

Whenever \(I(f)<\infty\), the score is defined \(f(x)\,dx\)-almost
everywhere by
\begin{equation}
   \rho_f(x)
   \triangleq
   \frac{\nabla f(x)}{f(x)}
   =
   \nabla\log f(x)
   \qquad\text{on }\{f>0\},
   \label{eq:score-definition}
\end{equation}
and we set \(\rho_f(x)=0\) on \(\{f=0\}\).
By \cite[Proposition~12.1]{Upper}, the weak Sobolev definition
\eqref{eq:weak-fisher} is equivalent to the usual score representation,
and hence
\begin{equation}
   I(f)
   =
   \int_{\R^d}
   \|\rho_f(x)\|^2 f(x)\,dx.
\end{equation}
The Fisher information matrix is
\begin{equation}
   \Jmat(f)=\int_{\R^d}\rho_f(x)\rho_f(x)^{\mathsf T}f(x)\,dx,
   \label{eq:fisher-matrix-definition}
\end{equation}
so that $I(f)=\tr\Jmat(f)$.

For $A=(a_{ij})\in\R^{d\times d}$, write
\[
   \|A\|_{\HS}^2=\tr(A^{\mathsf T}A)=\sum_{i,j=1}^d a_{ij}^2,
   \qquad
   \|A\|_{\op}=\sup_{\|x\|=1}\|Ax\|.
\]
For a smooth positive density $f$, define the second-order Fisher production
by
\begin{equation}
   K(f)=\int_{\R^d}\|\nabla^2\log f(x)\|_{\HS}^2f(x)\,dx.
   \label{eq:K-definition}
\end{equation}
When $d=1$,
\begin{equation}
   \rho_f=(\log f)',
   \qquad
   I(f)=\int_{\R}\rho_f(x)^2f(x)\,dx,
   \qquad
   K(f)=\int_{\R}|\rho_f'(x)|^2f(x)\,dx.
\end{equation}
The functional $K$ is denoted by $J$ in
Toscani~\cite{toscani2015strengthened}; we reserve $\Jmat$ for the Fisher
information matrix. For possibly nonsmooth original laws, \(K\) is used only after Gaussian smoothing.

\subsection{Conditional information functionals}
\label{sec:conditional-information-functionals}

Let $(X,Y)$ admit a conditional density $f_{X\mid Y}(\cdot\mid y)$. For
$\Phi\in\{h,I,K\}$, use the pointwise convention
\begin{equation}
   \Phi(X\mid Y=y)=\Phi\bigl(f_{X\mid Y}(\cdot\mid y)\bigr).
\end{equation}
We also set
\begin{equation}
   \tau_{X\mid Y=y}(R)
   =\tau_{f_{X\mid Y}(\cdot\mid y)}(R).
\end{equation}
Thus $\Phi(X\mid Y)$ and $\tau_{X\mid Y}(R)$ are random variables determined
by $Y$, while $\E\Phi(X\mid Y)$ denotes the corresponding average. The same
convention is used for conditional means and covariance matrices. In
particular, $\Jmat(X\mid Y=y)$ denotes the Fisher information matrix of the
conditional law, and
\begin{equation}
   I(X\mid Y=y)=\tr\Jmat(X\mid Y=y).
\end{equation}

\subsection{Conditional relative Fisher information and Gaussian reference}
\label{sec:relative-definitions}

Let $\Sigma\succ0$ and let $\gamma_\Sigma=N(0,\Sigma)$. If $(X,Y)$ has a
conditional density and $\E[X\mid Y]=0$ almost surely, define the averaged
conditional relative entropy with respect to the fixed Gaussian reference by
\begin{equation}
   \mathcal D_\Sigma(X\mid Y)
   =\E D\!\left(P_{X\mid Y}\middle\|\gamma_\Sigma\right),
   \label{eq:conditional-relative-entropy-def}
\end{equation}
where $D(P\|Q)$ denotes relative entropy. Whenever the conditional score is
square-integrable, define the dimensionless averaged conditional relative
Fisher information by
\begin{equation}
   \mathcal J_\Sigma(X\mid Y)
   =\E\int_{\R^d}
   \left\|\Sigma^{1/2}
   \bigl(\rho_{X\mid Y}(x)+\Sigma^{-1}x\bigr)\right\|^2
   f_{X\mid Y}(x\mid Y)\,dx.
   \label{eq:conditional-relative-fisher-def}
\end{equation}
The reference covariance is fixed rather than equal to the random conditional
covariance. This is natural here because
\[
   \E\Cov(S_n\mid\etaVec_n)=\Sigma
\]
for every $n$.

\subsection{Log-concave densities}

A probability density $f$ on $\R^d$ is log-concave if
\begin{equation}
   f(x)=e^{-\psi(x)},
\end{equation}
where $\psi:\R^d\to(-\infty,+\infty]$ is a proper lower-semicontinuous convex
function. Equivalently,
\begin{equation}
   f((1-\lambda)x+\lambda y)
   \ge f(x)^{1-\lambda}f(y)^\lambda,
   \qquad x,y\in\R^d,\quad\lambda\in[0,1].
\end{equation}
The conditional law of $X$ given $Y$ is called log-concave if
$f_{X\mid Y}(\cdot\mid y)$ is log-concave for $P_Y$-almost every $y$.
Log-concavity is preserved under affine transformations, products,
convolution, and marginalization; in particular, Gaussian smoothing preserves
log-concavity \cite{saumard2014log}.

\subsection{Gaussian smoothing and scaling}

Let $X$ have finite second moment and let $G\sim N(0,\Id)$ be independent of
$X$. For $t\ge0$, define
\begin{equation}
   X_t=X+\sqrt t\,G.
   \label{eq:additive-smoothing}
\end{equation}
We also use the Ornstein--Uhlenbeck smoothing
\begin{equation}
   P_t^*X=e^{-t}X+\sqrt{1-e^{-2t}}\,G.
\end{equation}
For $t>0$, both smoothed laws have strictly positive smooth densities.
Along the additive heat flow, de Bruijn's identity
\cite[Sec.~II-D]{rioul2011information} and the Fisher-information
dissipation identity \cite{toscani2015reciprocal} give
\begin{equation}
   \frac{d}{dt}h(X_t)=\frac12 I(X_t),
   \label{eq:debruijn}
\end{equation}
and
\begin{equation}
   \frac{d}{dt}I(X_t)=-K(X_t).
   \label{eq:fisher-dissipation}
\end{equation}

\section{Main results}
\label{sec:main-results}

Let $m\ge1$ be fixed, let $\eta\in\R^m$, and let
$\{(\xi_i,\eta_i)\}_{i\ge1}$ be independent copies of $(\xi,\eta)$. Define
\begin{equation}
   Z_i=\xi_i-\E[\xi_i\mid\eta_i],
   \qquad
   S_n=\frac1{\sqrt n}\sum_{i=1}^n Z_i,
   \qquad
   \etaVec_n=(\eta_1,\ldots,\eta_n).
   \label{eq:centered-normalized-sum}
\end{equation}
Conditionally on $\etaVec_n$, the centering term is deterministic; hence
translation invariance makes the centered and uncentered conditional
Fisher-information problems equivalent.

\subsection{The fixed-dimensional criterion}
\label{sec:multivariate-result}

Let $\xi\in\R^d$, where $d\ge1$ is fixed, and define
\begin{equation}
   \Sigma=\E\Cov(\xi\mid\eta).
   \label{eq:multi-average-covariance}
\end{equation}
Assume that $\Sigma$ is finite and positive definite. Put
\begin{equation}
   X_i=\Sigma^{-1/2}Z_i,
   \qquad
   W_n=\Sigma^{-1/2}S_n
   =\frac1{\sqrt n}\sum_{i=1}^n X_i.
   \label{eq:multi-standardized-sum}
\end{equation}
Then
\begin{equation}
   \E[X_i\mid\eta_i]=0,
   \qquad
   \E\Cov(X_i\mid\eta_i)=\Id.
   \label{eq:standardized-moments}
\end{equation}

\begin{theorem}[Fixed-dimensional conditional Fisher-information criterion]
\label{thm:main}
Assume that, for $P_\eta$-almost every $y\in\R^m$, the conditional law of
$\xi$ given $\eta=y$ admits a full-dimensional log-concave Lebesgue density on
$\R^d$. Define
\begin{equation}
   \mathcal I_n=\E I(W_n\mid\etaVec_n).
   \label{eq:multi-averaged-fisher}
\end{equation}
Then
\begin{equation}
   \mathcal I_n\longrightarrow d
   \label{eq:multi-fisher-limit}
\end{equation}
if and only if there exists $n_0\ge1$ such that
\begin{equation}
   \mathcal I_{n_0}<\infty.
   \label{eq:multi-finite-level-fisher}
\end{equation}
\end{theorem}
The proof of Theorem~\ref{thm:main}
is given in Section~\ref{sec:multivariate-proof}.

\begin{corollary}[One-dimensional criterion]
\label{cor:scalar}
Let $\xi\in\R$ and define
\begin{equation}
   \sigma^2=\E\operatorname{Var}(\xi\mid\eta),
   \qquad 0<\sigma^2<\infty.
   \label{eq:average-conditional-variance}
\end{equation}
Assume that the conditional law of $\xi$ given $\eta$ has a log-concave
density almost surely. If
\begin{equation}
   I_n=\E I(S_n\mid\etaVec_n),
   \label{eq:scalar-averaged-fisher}
\end{equation}
then
\begin{equation}
   I_n\longrightarrow\frac1{\sigma^2}
   \label{eq:logconcave-fisher-limit}
\end{equation}
if and only if $I_{n_0}<\infty$ for some $n_0\ge1$.
\end{corollary}

\begin{proof}
Apply Theorem~\ref{thm:main} with $d=1$. Then $\Sigma=\sigma^2$,
$W_n=S_n/\sigma$, and the scaling identity gives
\[
   \E I(W_n\mid\etaVec_n)=\sigma^2\E I(S_n\mid\etaVec_n).
\]
The finite-level conditions scale in the same way.
\end{proof}

\begin{remark}
The finite-level condition does not require
$\E I(\xi\mid\eta)<\infty$. It allows nonsmooth log-concave conditional laws,
including compactly supported densities, provided Fisher information becomes
finite after finitely many convolutions.
\end{remark}

\begin{example}[A genuinely finite-convolution-level model]
\label{ex:state-dependent-uniform}
Let $\eta$ be Bernoulli and, conditionally on $\eta=j$, let
$\xi\sim\operatorname{Unif}[-a_j,a_j]$, where $0<a_0,a_1<\infty$. Each
one-step conditional density is log-concave but has infinite weak Fisher
information because of its jumps at the endpoints. On the other hand, for
every realization of $(\eta_1,\eta_2,\eta_3)$, the density of
$\xi_1+\xi_2+\xi_3$ is a compactly supported piecewise-quadratic spline that
vanishes quadratically at the two boundary points. Hence its weak derivative
satisfies $f'^2/f\in L^1$, and its Fisher information is finite.  
see also the classical triple-convolution result of Bobkov, Chistyakov, and G"otze 
\cite{bobkov2014fisher}. 
Since there are only finitely many state configurations, the scaling identity yields
\[
   \E I(S_3\mid\etaVec_3)<\infty.
\]
Thus Corollary~\ref{cor:scalar} applies with $n_0=3$, although the original
conditional Fisher information is infinite almost surely.
\end{example}

\begin{remark}[Relation to the entropic conditional CLT and an alternative scalar proof]
Whenever $I_n<\infty$, weak integration by parts and
\[
   \E \operatorname{Var}(S_n\mid\etaVec_n)=\sigma^2
\]
give
\begin{equation}
   \mathcal J_{\sigma^2}(S_n\mid\etaVec_n)
   =
   \sigma^2 I_n-1.
   \label{eq:scalar-relative-fisher-identity}
\end{equation}
Thus Corollary~\ref{cor:scalar} is equivalently a convergence
result for the conditional relative Fisher information. As shown in
Corollary~\ref{cor:relative-information} below, this Fisher convergence
implies the corresponding conditional relative-entropy and entropy
convergence.

This implication should be understood within the present log-concave
setting. In contrast, the finite-entropy criterion of
\cite{ye2026finiteentropy} establishes the entropic conditional CLT under
the weaker structural assumption of absolute continuity, without
log-concavity. Thus the present scalar result gives a stronger mode of
convergence, while relying on the additional log-concavity assumption.

Although Corollary~\ref{cor:scalar} follows directly from
Theorem~\ref{thm:main} by setting $d=1$,
Appendix~\ref{app:scalar-alternative} gives an independent scalar proof
along the entropic route of Ma et al.~\cite{ma2024entropic} and
\cite{ye2026finiteentropy}. The additional ingredient in that proof is
the continuity of the second-order Fisher production $K$ established in
Appendix~\ref{app:smoothed-continuity}.
\end{remark}

\subsection{Fisher-matrix and relative-information consequences}

\begin{corollary}[Averaged Fisher matrix]
\label{cor:matrix}
Under the assumptions of Theorem~\ref{thm:main}, if the equivalent conditions
of that theorem hold, then
\begin{equation}
   \E\Jmat(S_n\mid\etaVec_n)\longrightarrow\Sigma^{-1}
   \label{eq:fisher-matrix-limit}
\end{equation}
in operator norm. Consequently,
\begin{equation}
   \E I(S_n\mid\etaVec_n)\longrightarrow\tr(\Sigma^{-1}).
   \label{eq:unstandardized-scalar-limit}
\end{equation}
\end{corollary}
The proof of Corollary~\ref{cor:matrix} is given in
Section~\ref{sec:fisher-matrix-proof}.

\begin{corollary}[Relative Fisher and entropic convergence]
\label{cor:relative-information}
Under the assumptions of Theorem~\ref{thm:main}, suppose the finite-level
condition holds. Then
\begin{equation}
   \mathcal J_\Sigma(S_n\mid\etaVec_n)
   =\tr\!\left(\Sigma\,\E\Jmat(S_n\mid\etaVec_n)\right)-d
   \longrightarrow0.
   \label{eq:relative-fisher-limit}
\end{equation}
Moreover,
\begin{equation}
   0\le\mathcal D_\Sigma(S_n\mid\etaVec_n)
   \le\frac12\mathcal J_\Sigma(S_n\mid\etaVec_n),
   \label{eq:relative-entropy-fisher-bound}
\end{equation}
and therefore
\begin{equation}
   \mathcal D_\Sigma(S_n\mid\etaVec_n)\longrightarrow0.
   \label{eq:relative-entropy-limit}
\end{equation}
Equivalently,
\begin{equation}
   \E h(S_n\mid\etaVec_n)
   \longrightarrow\frac12\log\!\left((2\pi e)^d\det\Sigma\right).
   \label{eq:conditional-entropy-limit-from-fisher}
\end{equation}
\end{corollary}

The proof of Corollary~\ref{cor:relative-information} is given in
Section~\ref{sec:relative-application-proof}, where the same relative Fisher
deficit controls a channel mutual-information gap.

\section{Proof of the fixed-dimensional criterion}
\label{sec:multivariate-proof}
In this section, we prove Theorem~\ref{thm:main} and Corollary~\ref{cor:matrix}.
Throughout this section, the notation is that of
Section~\ref{sec:multivariate-result}, and the dimension $d$ is fixed.

\subsection{First-order Fisher inequalities}
If $U,V\in\R^d$ are independent with finite Fisher information and
$a,b\ne0$ satisfy $a^2+b^2=1$, then \cite{dembo1991information}
\begin{equation}
   I(aU+bV)\le a^2I(U)+b^2I(V).
   \label{eq:weighted-fisher}
\end{equation}
Moreover, if $U$ has finite Fisher information and $V$ is independent of
$U$, then, for $a\ne0$, \cite[Corollary~14.3]{Upper} gives
\begin{equation}
   I(aU+bV)
   \le I(aU)
   =\frac1{a^2}I(U).
   \label{eq:onecomponent}
\end{equation}
We apply these inequalities pointwise to regular conditional laws.

\subsection{Conditional multivariate Lindeberg--Feller convergence}
Set
\[
   C_i=\Cov(X_i\mid\eta_i)=\E[X_iX_i^{\mathsf T}\mid\eta_i],
   \qquad
   V_n=\Cov(W_n\mid\etaVec_n)=\frac1n\sum_{i=1}^n C_i,
\]
and, for $\varepsilon>0$,
\[
   \Lambda_n(\varepsilon)
   =\frac1n\sum_{i=1}^n
   \E\!\left[
      \|X_i\|^2\mathbf1_{\{\|X_i\|>\varepsilon\sqrt n\}}
      \middle|\eta_i
   \right].
\]
Let $\mathcal P(\mathbb R^d)$ denote the set of all Borel probability
measures on $\mathbb R^d$. For a bounded Lipschitz function
$f:\mathbb R^d\to\mathbb R$, define
\[
   \|f\|_{\mathrm{BL}}
   \triangleq
   \|f\|_\infty+\operatorname{Lip}(f),
\]
where
\[
   \operatorname{Lip}(f)
   \triangleq
   \sup_{x\ne y}
   \frac{|f(x)-f(y)|}{\|x-y\|}.
\]
For $\mu,\nu\in\mathcal P(\mathbb R^d)$, define the bounded-Lipschitz
metric by
\[
   d_{\mathrm{BL}}(\mu,\nu)
   \triangleq
   \sup_{\|f\|_{\mathrm{BL}}\le1}
   \left|
      \int_{\mathbb R^d} f\,d\mu
      -
      \int_{\mathbb R^d} f\,d\nu
   \right|.
\]
This metric metrizes weak convergence on
$\mathcal P(\mathbb R^d)$; see, e.g.,
\cite[Theorem~11.3.3]{Dudley2002}.
For random probability measures, we use the phrase
\emph{weakly in probability} to mean convergence in probability with
respect to $d_{\mathrm{BL}}$.

Set
\[
   \nu_n=\mathcal L(W_n\mid\etaVec_n),
   \qquad
   \gamma_d=N(0,I_d).
\]
\begin{proposition}[Conditional multivariate CLT]
\label{prop:conditional-clt}
Under the assumptions of Theorem~\ref{thm:main},
\[
   d_{\mathrm{BL}}(\nu_n,\gamma_d)\longrightarrow0
   \qquad\text{in probability}.
\]
Equivalently,
\[
   \Law(W_n\mid\etaVec_n)\Longrightarrow N(0,\Id)
   \qquad\text{weakly in probability}.
\]
\end{proposition}

\begin{proof}
See Appendix~\ref{app:conditional-lindeberg-proof}.
\end{proof}

\subsection{Identification of the Gaussian-smoothed limits}
Fix $t>0$ and define
\[
   W_{n,t}=W_n+\sqrt t\,G,
\]
where $G\sim N(0,\Id)$ is independent. Conditionally on $\etaVec_n$, let
$p_{n,t}$ denote the density of $W_{n,t}$.

\begin{lemma}[Gaussian convolution: derivative convergence]
\label{lem:C2conv}
Let $\mu_n$ be probability measures on $\R^d$ such that
$d_{\mathrm{BL}}(\mu_n,\mu)\to0$. For every fixed $t>0$ and multi-index
$\alpha$ with $|\alpha|\le2$,
\[
   D^\alpha(\mu_n*\phi_t)\longrightarrow D^\alpha(\mu*\phi_t)
\]
locally uniformly on $\R^d$, where $\phi_t$ is the $N(0,t\Id)$ density.
\end{lemma}

\begin{proof}
See Appendix~\ref{app:gaussian-smoothing-proof}.
\end{proof}

\begin{lemma}[Posterior-covariance Hessian identity]
\label{lem:hessian}
Let $X$ be any $\R^d$-valued random vector with finite second moment and let
\[
   Y=X+\sqrt t\,G,
   \qquad G\sim N(0,\Id),
\]
with $G$ independent of $X$.  Let $p_t$ be the density of $Y$, and define
\[
   m(y):=\E[X\mid Y=y].
\]
Then
\begin{align}
   \nabla\log p_t(y)
   &=\frac{m(y)-y}{t},
   \label{eq:tweedie}\\
   \nabla^2\log p_t(y)
   &=-\frac1t\Id
   +\frac1{t^2}\Cov(X\mid Y=y).
   \label{eq:hessianposterior}
\end{align}
If the law of $X$ is log-concave, then
\begin{equation}
   -\frac1t\Id
   \preceq
   \nabla^2\log p_t(y)
   \preceq0,
   \label{eq:hessianbound}
\end{equation}
and hence
\begin{equation}
   K(Y)\le\frac d{t^2}.
   \label{eq:Kbound}
\end{equation}
\end{lemma}

\begin{proof}
See Appendix~\ref{app:hessian-proof}.
\end{proof}

\begin{proposition}[Identification of the smoothed $K$ limit]
\label{prop:Klimit}
Under the assumptions of Theorem~\ref{thm:main}, for every fixed $t>0$,
we have
\[
   K(W_{n,t}\mid\etaVec_n)
   \longrightarrow
   \frac{d}{(1+t)^2}
   \qquad\text{in probability},
\]
and therefore
\begin{equation}
   \mathcal K_n(t)
   :=\E K(W_{n,t}\mid\etaVec_n)
   \longrightarrow
   \frac{d}{(1+t)^2}.
   \label{eq:Kexpectlimit}
\end{equation}
\end{proposition}

\begin{proof}
By Proposition~\ref{prop:conditional-clt}, we know that
\begin{equation}
   d_{\mathrm{BL}}(\nu_n,\gamma_d)
   \overset{P}{\longrightarrow}0.
\end{equation}
Therefore, for any subsequence $\{n_k\}$, there exists a further subsequence
$\{n_{k_j}\}$ such that
\[
d_{\mathrm{BL}}(\nu_{n_{k_j}},\gamma_d)
\longrightarrow 0
\qquad \text{a.s.}
\]
Fix an environment $\omega$ in the corresponding probability-one event.
For notational simplicity, we suppress the dependence on $\omega$ in what
follows.

By Lemma~\ref{lem:C2conv}, the Gaussian-smoothed conditional densities
satisfy
\[
   p_{n_{k_j},t}\to\phi_{1+t},
   \qquad
   \nabla p_{n_{k_j},t}\to\nabla\phi_{1+t},
   \qquad
   \nabla^2p_{n_{k_j},t}\to\nabla^2\phi_{1+t}
\]
locally uniformly. On each fixed ball $B_R$, the limiting Gaussian
density has a strictly positive minimum. Hence, for all sufficiently large $j$,
\[
   \nabla^2\log p_{n_{k_j},t}
   =
   \frac{\nabla^2p_{n_{k_j},t}}{p_{n_{k_j},t}}
   -
   \frac{
      \nabla p_{n_{k_j},t}
      \nabla p_{n_{k_j},t}^{\mathsf T}
   }{
      p_{n_{k_j},t}^2
   }
\]
converges uniformly on $B_R$ to
\[
   \nabla^2\log\phi_{1+t}
   =
   -\frac1{1+t}\Id.
\]
Therefore,
\[
   \int_{B_R}
   \|\nabla^2\log p_{n_{k_j},t}\|_{\HS}^2
   p_{n_{k_j},t}
   \longrightarrow
   \int_{B_R}
   \frac d{(1+t)^2}\phi_{1+t}.
\]

For the tails, Lemma~\ref{lem:hessian} gives the deterministic
pointwise bound
\[
   \|\nabla^2\log p_{n_{k_j},t}\|_{\HS}^2
   \le
   \frac d{t^2}.
\]
Thus,
\[
   \int_{B_R^c}
   \|\nabla^2\log p_{n_{k_j},t}\|_{\HS}^2
   p_{n_{k_j},t}
   \le
   \frac d{t^2}
   \Pp\left(
      W_{n_{k_j},t}\notin B_R
      \mid
      \etaVec_{n_{k_j}}
   \right).
\]
By Proposition~\ref{prop:conditional-clt} and continuity of Gaussian convolution, we have 
\begin{equation}
   \Law(W_{n_{k_{j}},t}\mid\etaVec_{n_{k_{j}}})
   \Longrightarrow N(0,(1+t)\Id).
\end{equation}
Hence, for every $R>0$ such that
$\partial B_R$ has zero Gaussian measure,
\[
\Pp\left(
   W_{n_{k_j},t}\notin B_R
   \mid
   \etaVec_{n_{k_j}}
\right)
\longrightarrow
\Pp\left(
   N(0,(1+t)\Id)\notin B_R
\right).
\]
Letting first $j\to\infty$ and then $R\to\infty$ gives
\[
   K(W_{n_{k_j},t}\mid\etaVec_{n_{k_j}})
   \longrightarrow
   K(N(0,(1+t)\Id))
   =
   \frac d{(1+t)^2}.
\]
Since every subsequence admits a further subsequence along which the
above almost-sure convergence holds, it follows that
\[
   K(W_{n,t}\mid\etaVec_n)
   \overset{P}{\longrightarrow}
   \frac d{(1+t)^2}
\]
for the original sequence.

Finally, from Lemma~\ref{lem:hessian}, we obtain
\begin{equation}\label{eq:K-bound}
     0
   \le
   K(W_{n,t}\mid\etaVec_n)
   \le
   \frac d{t^2}, \qquad n\ge1.
\end{equation}
Convergence in
probability together with this uniform bound implies convergence in
$L^1$, which yields \eqref{eq:Kexpectlimit}.
\end{proof}

\begin{proposition}[Identification of the smoothed Fisher limit]
\label{prop:Ilimit}
Under the assumptions of Theorem~\ref{thm:main}, for every fixed $t>0$,
we have
\begin{equation}
   \mathcal I_n(t)
   :=\E I(W_{n,t}\mid\etaVec_n)
   \longrightarrow
   \frac d{1+t}.
   \label{eq:Iexpectlimit}
\end{equation}
\end{proposition}

\begin{proof}
Fix a conditional law.  By the Fisher dissipation identity
\eqref{eq:fisher-dissipation}, for $T>t$,
\[
   I(W_{n,t}\mid\etaVec_n)
   -I(W_{n,T}\mid\etaVec_n)
   =\int_t^T K(W_{n,s}\mid\etaVec_n)\,ds.
\]
By \eqref{eq:onecomponent}, we obtain 
\[
   I(W_{n,T}\mid\etaVec_n)
   \le I(\sqrt T\,G)
   =\frac dT.
\]
Hence $I(W_{n,T}\mid\etaVec_n)\to0$ as $T\to\infty$, and therefore
\[
   I(W_{n,t}\mid\etaVec_n)
   =\int_t^\infty K(W_{n,s}\mid\etaVec_n)\,ds.
\]
Taking expectations and using Tonelli's theorem \cite{Kallenberg2021}  gives
\[
   \mathcal I_n(t)=\int_t^\infty\mathcal K_n(s)\,ds.
\]
For $s\ge t$, Lemma~\ref{lem:hessian} gives
\[
   0\le\mathcal K_n(s)\le\frac d{s^2},
\]
and $d/s^2$ is integrable on $[t,\infty)$.  Proposition~\ref{prop:Klimit}
and dominated convergence therefore yield
\[
   \lim_{n\to\infty}\mathcal I_n(t)
   =\int_t^\infty\frac d{(1+s)^2}\,ds
   =\frac d{1+t}.
\]
\end{proof}

\subsection{Proof of Theorem~\ref{thm:main}}

\begin{proof}[Proof of Theorem~\ref{thm:main}]

\emph{Necessity.}
If
\[
\mathcal I_n\longrightarrow d<\infty,
\]
then $\mathcal I_n<\infty$ for all sufficiently large $n$. Hence there
exists $n_0\ge1$ such that
\[
\mathcal I_{n_0}<\infty.
\]

\medskip
\noindent
\emph{Sufficiency.}
Assume that
\[
\mathcal I_{n_0}<\infty
\]
for some $n_0\ge1$.

We first note that the finite-Fisher assumption propagates to every
$n\ge n_0$. For $n>n_0$, condition on $\etaVec_n$ and split
\[
W_n
=
\sqrt{\frac{n_0}{n}}\,W_{n_0}^{(1)}
+
\sqrt{\frac{n-n_0}{n}}\,R_{n-n_0},
\]
where the two blocks are conditionally independent. By the classical
one-component Fisher inequality \eqref{eq:onecomponent},
\[
I(W_n\mid\etaVec_n)
\le
\frac{n}{n_0}
I(W_{n_0}^{(1)}\mid\etaVec_{n_0}^{(1)}).
\]
Taking expectations gives
\begin{equation}
\mathcal I_n
\le
\frac{n}{n_0}\mathcal I_{n_0}
<\infty,
\qquad n\ge n_0.
\label{eq:multi-eventual-finiteness}
\end{equation}

We next establish convergence along the dyadic subsequence. Set
\[
N_m:=2^m n_0,
\qquad m\ge0,
\]
and, for $t>0$, define
\[
\mathcal K_m(t)
:=
\E K(W_{N_m}+\sqrt t\,G\mid\etaVec_{N_m}).
\]

Lemma~\ref{lem:4.2}, applied with
$\lambda=1/\sqrt2$, gives, for independent log-concave random vectors
$U,V\in\R^d$ satisfying the regularity assumptions there,
\begin{equation}
K\left(\frac{U+V}{\sqrt2}\right)
\le
\frac{K(U)+K(V)}2.
\label{eq:multi-Kmean-direct}
\end{equation}

Split the $N_{m+1}=2N_m$ summands into two conditionally independent
blocks. Let $G_1,G_2$ be independent standard Gaussian vectors,
independent of all other random variables. Conditionally on the two
environment blocks,
\[
W_{N_{m+1}}+\sqrt t\,G
\overset{d}{=}
\frac{
(W_{N_m}^{(1)}+\sqrt t\,G_1)
+
(W_{N_m}^{(2)}+\sqrt t\,G_2)
}{\sqrt2}.
\]
The two Gaussian-smoothed conditional block densities are log-concave.
Applying \eqref{eq:multi-Kmean-direct} pointwise in the conditioning
variables and then taking expectations yields
\begin{equation}
\mathcal K_{m+1}(t)
\le
\mathcal K_m(t),
\qquad t>0.
\label{eq:dyadic-K-monotone-multi}
\end{equation}
Consequently,
\begin{equation}
0\le
\mathcal K_m(t)
\le
\mathcal K_0(t),
\qquad m\ge0,\quad t>0.
\label{eq:dyadic-K-domination}
\end{equation}

We now remove the Gaussian smoothing. By
\cite[Corollary~14.3]{Upper}, if $I(X)<\infty$, then
\begin{equation}
I(X+\sqrt s\,G)
\longrightarrow
I(X)
\qquad\text{as }s\downarrow0.
\label{eq:fisher-gaussian-zero-limit}
\end{equation}
On the other hand, the Fisher dissipation identity
\eqref{eq:fisher-dissipation} gives, for $0<\varepsilon<t$,
\[
I(X+\sqrt{\varepsilon}\,G)
-
I(X+\sqrt t\,G)
=
\int_{\varepsilon}^{t}
K(X+\sqrt s\,G)\,ds.
\]
Letting $\varepsilon\downarrow0$ and using
\eqref{eq:fisher-gaussian-zero-limit}, we obtain
\begin{equation}
I(X)-I(X+\sqrt t\,G)
=
\int_0^t K(X+\sqrt s\,G)\,ds
\label{eq:multi-endpoint-dissipation}
\end{equation}
whenever $I(X)<\infty$.

By \eqref{eq:multi-eventual-finiteness},
$\mathcal I_{N_m}<\infty$ for every $m$. Hence
\[
I(W_{N_m}\mid\etaVec_{N_m})<\infty
\qquad\text{a.s.},
\]
so \eqref{eq:multi-endpoint-dissipation} can be applied to the
conditional law of $W_{N_m}$. Taking expectations gives
\begin{equation}
\mathcal I_{N_m}
=
\mathcal I_{N_m}(t)
+
\int_0^t\mathcal K_m(s)\,ds.
\label{eq:dyadic-dissipation}
\end{equation}

For $m=0$, the same identity gives
\[
\int_0^t\mathcal K_0(s)\,ds
=
\mathcal I_{n_0}
-
\mathcal I_{n_0}(t)
\le
\mathcal I_{n_0}
<\infty.
\]
Thus $\mathcal K_0$ is an integrable dominating function on $(0,t)$
for the sequence $\{\mathcal K_m\}_{m\ge0}$.

For every fixed $s>0$, Proposition~\ref{prop:Klimit} gives
\begin{equation}\label{eq:K-limit}
\mathcal K_m(s)
=
\E K(W_{N_m}+\sqrt s\,G\mid\etaVec_{N_m})
\longrightarrow
\frac{d}{(1+s)^2} 
\qquad\text{as } m\to\infty.
\end{equation}
Therefore, by  \eqref{eq:K-limit},\eqref{eq:dyadic-K-domination} and the dominated
convergence theorem,
\begin{equation}
\int_0^t\mathcal K_m(s)\,ds
\longrightarrow
\int_0^t\frac{d}{(1+s)^2}\,ds.
\label{eq:dyadic-K-integral-limit}
\end{equation}
At the same time, Proposition~\ref{prop:Ilimit} yields
\begin{equation}
\mathcal I_{N_m}(t)
\longrightarrow
\frac{d}{1+t}.
\label{eq:dyadic-smoothed-I-limit}
\end{equation}

Passing to the limit $m\to\infty$ in
\eqref{eq:dyadic-dissipation} and using
\eqref{eq:dyadic-K-integral-limit} and
\eqref{eq:dyadic-smoothed-I-limit}, we obtain
\begin{align}
\lim_{m\to\infty}\mathcal I_{N_m}
&=
\frac{d}{1+t}
+
\int_0^t\frac{d}{(1+s)^2}\,ds
\notag\\
&=
\frac{d}{1+t}
+
d\left(1-\frac1{1+t}\right)
\notag\\
&=
d.
\label{eq:multi-dyadic-final}
\end{align}
Thus
\begin{equation}
\lim_{m \rightarrow \infty}\mathcal I_{2^m n_0}
= d.
\label{eq:multi-dyadic-subsequence}
\end{equation}

It remains to pass from the dyadic subsequence to the full sequence.
Condition on the full environment and split $W_{m+n}$ into its first
$m$ and last $n$ normalized blocks. Applying the Blachman--Stam
inequality \eqref{eq:weighted-fisher} pointwise in the conditioning
variables and then taking expectations gives
\[
\mathcal I_{m+n}
\le
\frac{m}{m+n}\mathcal I_m
+
\frac{n}{m+n}\mathcal I_n.
\]
Equivalently, the eventually finite sequence
\[
a_n:=n\mathcal I_n
\]
is subadditive:
\[
a_{m+n}\le a_m+a_n.
\]
Fekete's lemma for extended-real subadditive sequences gives the existence of
\[
   \lim_{n\to\infty}\mathcal I_n
\]
(the finitely many possibly infinite initial terms are irrelevant). Hence, by \eqref{eq:multi-dyadic-subsequence}, we obtain
\begin{equation}
   \lim_{n\to\infty}\mathcal I_n=d.
\end{equation}

\end{proof}

\subsection{Matrix Cram\'er--Rao and inversion Jensen inequalities}

The following matrix inequalities are needed only for the Fisher-matrix
corollary, not for identifying the scalar subadditive limit.

\begin{lemma}[Matrix Cram\'er--Rao inequality]
\label{lem:cramerrao}
Let $Y\in\R^d$ have a full-dimensional density, finite covariance
matrix $V\succ0$, and finite Fisher information. Then
\[
   \Jmat(Y)\succeq V^{-1}.
\]
\end{lemma}

\begin{proof}
This is the multivariate covariance--Fisher information inequality;
see Kagan and Landsman~\cite{KaganLandsman1999}.
\end{proof}

\subsection{Fisher matrix convergence}
\label{sec:fisher-matrix-proof}
\begin{proof}[Proof of Corollary~\ref{cor:matrix}]
Since $W_n=\Sigma^{-1/2}S_n$, we obtain
\begin{equation}\label{eq:matrix-transform-SW}
      \Jmat(W_n\mid\etaVec_n)
   =\Sigma^{1/2}
    \Jmat(S_n\mid\etaVec_n)
    \Sigma^{1/2}.
\end{equation}
Let
\[
   A_n
   :=
   \E\Jmat(S_n\mid\etaVec_n).
\]
By Theorem~\ref{thm:main},
\[
   \mathcal I_n
   =
   \E I(W_n\mid\etaVec_n)
   =
   \E\tr\Jmat(W_n\mid\etaVec_n)
   \longrightarrow d.
\]
Hence $\mathcal I_n<\infty$ for all sufficiently large $n$.

Since $\Jmat(W_n\mid\etaVec_n)$ is positive semidefinite, using
\eqref{eq:matrix-transform-SW} and the cyclicity of the trace gives
\begin{align}
   \tr\Jmat(S_n\mid\etaVec_n)
   &=
   \tr\left(
      \Sigma^{-1/2}
      \Jmat(W_n\mid\etaVec_n)
      \Sigma^{-1/2}
   \right)
   \notag\\
   &=
   \tr\left(
      \Sigma^{-1}
      \Jmat(W_n\mid\etaVec_n)
   \right)
   \notag\\
   &\le
   \|\Sigma^{-1}\|_{\op}
   \tr\Jmat(W_n\mid\etaVec_n).
\end{align}
Indeed, for positive-semidefinite matrices $C$ and $D$,
\[
   \tr(CD)
   \le
   \|C\|_{\op}\tr(D).
\]
Taking expectations therefore yields
\begin{align}
   \E\tr\Jmat(S_n\mid\etaVec_n)
   &\le
   \|\Sigma^{-1}\|_{\op}
   \E\tr\Jmat(W_n\mid\etaVec_n)
   \notag\\
   &=
   \|\Sigma^{-1}\|_{\op}\mathcal I_n
   <\infty.
   \label{eq:An-finite-trace}
\end{align}
Since $\Jmat(S_n\mid\etaVec_n)$ is positive semidefinite,
\[
\left|\Jmat(S_n\mid\etaVec_n)_{ij}\right|
\le
\tr\Jmat(S_n\mid\etaVec_n).
\]
Hence \eqref{eq:An-finite-trace} implies that every entry of
$\Jmat(S_n\mid\etaVec_n)$ is integrable. Therefore
\[
   A_n
   =
   \E\Jmat(S_n\mid\etaVec_n)
\]
is well defined and has finite entries for all sufficiently large $n$.
Applying Lemma~\ref{lem:cramerrao} to the conditional law of
\(S_n\) given \(\etaVec_n\), we have, almost surely,
\begin{equation}\label{eqf:CR}
   \Jmat(S_n\mid\etaVec_n)
   \succeq
   \Cov(S_n\mid\etaVec_n)^{-1}.
\end{equation}
Taking traces and using \eqref{eq:An-finite-trace},
\[
   \E\tr\!\left(
      \Cov(S_n\mid\etaVec_n)^{-1}
   \right)
   \le
   \E\tr\Jmat(S_n\mid\etaVec_n)
   <\infty.
\]
Hence
\[
   \E\left[
      \Cov(S_n\mid\etaVec_n)^{-1}
   \right]
\]
is well defined. Taking expectations in \eqref{eqf:CR}
therefore gives
\[
   A_n
   \succeq
   \E\left[
      \Cov(S_n\mid\etaVec_n)^{-1}
   \right].
\]
 Since the inversion map $V\mapsto V^{-1}$ is operator convex on the
positive-definite cone, Jensen's operator inequality
\cite{HansenPedersen2003} gives
\[
   \E\left[
      \Cov(S_n\mid\etaVec_n)^{-1}
   \right]
   \succeq
   \left(
      \E\Cov(S_n\mid\etaVec_n)
   \right)^{-1}.
\]
Since
\[
   \E\Cov(S_n\mid\etaVec_n)=\Sigma,
\]
we obtain
\begin{equation}
   A_n\succeq\Sigma^{-1}.
   \label{eq:matrixlower}
\end{equation}
Theorem~\ref{thm:main} gives
\[
   \tr(\Sigma A_n)\longrightarrow d.
\]
Define
\[
   B_n:=\Sigma^{1/2}(A_n-\Sigma^{-1})\Sigma^{1/2}.
\]
By \eqref{eq:matrixlower}, $B_n\succeq0$, and
\[
   \tr(B_n)
   =\tr(\Sigma A_n)-d
   \longrightarrow0.
\]
For a positive-semidefinite matrix,
\[
   \|B_n\|_{\op}\le\tr(B_n),
\]
so $\|B_n\|_{\op}\to0$.  Since
\[
   A_n-\Sigma^{-1}
   =\Sigma^{-1/2}B_n\Sigma^{-1/2},
\]
we obtain
\[
   \|A_n-\Sigma^{-1}\|_{\op}
   \le
   \|\Sigma^{-1/2}\|_{\op}^2\|B_n\|_{\op}
   \longrightarrow0.
\]
Therefore
\[
   \E\Jmat(S_n\mid\etaVec_n)\longrightarrow\Sigma^{-1}
\]
in operator norm.  Taking traces gives
\[
   \E I(S_n\mid\etaVec_n)\longrightarrow\tr(\Sigma^{-1}).
\]

\end{proof}

\section{Information-theoretic consequences for state-dependent aggregate noise}
\label{sec:applications}
This section develops information-theoretic consequences of the
conditional Fisher information CLT for state-dependent aggregate
noise. We first show that the Fisher-matrix limit implies vanishing
conditional relative Fisher information and, through the Gaussian
logarithmic Sobolev inequality, vanishing conditional relative
entropy. We then derive two operational consequences: a low-SNR slope
formula for fixed finite-constellation signaling, governed directly
by the averaged conditional Fisher information matrix, and a
Gaussian-signaling result at fixed signal covariance.

In both operational results, \(\etaVec_n\) is interpreted as side
information available at the receiver about the states of the
individual noise components, while the normalization by \(n^{-1/2}\)
ensures that
\[
   \E\Cov(S_n\mid\etaVec_n)=\Sigma
   \qquad\text{for every }n.
\]

\subsection{Relative Fisher information controls conditional Gaussianity}
\label{sec:relative-application-proof}

\begin{proof}[Proof of Corollary~\ref{cor:relative-information}]
For every sufficiently large \(n\), the Fisher matrix is integrable by
Corollary~\ref{cor:matrix}.  Applying the weak integration-by-parts
identity to the conditional law of \(S_n\) given \(\etaVec_n\) gives
\begin{equation}
   \E\!\left[
      \rho_{S_n\mid\etaVec_n}(S_n)S_n^{\mathsf T}
      \,\middle|\,\etaVec_n
   \right]
   =-\Id
   \qquad\text{a.s.}
   \label{eq:conditional-score-ibp}
\end{equation}
Here the identity may be obtained by applying weak integration by parts
to cutoff approximations of the coordinate functions.

Expanding the square in the definition
\eqref{eq:conditional-relative-fisher-def}, using
\eqref{eq:conditional-score-ibp}, and then taking expectations yields
\begin{align}
   \mathcal J_\Sigma(S_n\mid\etaVec_n)
   &=
   \tr\!\left(
      \Sigma\,\E\Jmat(S_n\mid\etaVec_n)
   \right)
   +
   \tr\!\left(
      \Sigma^{-1}
      \E\Cov(S_n\mid\etaVec_n)
   \right)
   -2d
   \notag\\
   &=
   \tr\!\left(
      \Sigma\,\E\Jmat(S_n\mid\etaVec_n)
   \right)-d,
   \label{eq:relative-fisher-expansion}
\end{align}
because
\[
   \E\Cov(S_n\mid\etaVec_n)=\Sigma.
\]
By Corollary~\ref{cor:matrix},
\[
   \E\Jmat(S_n\mid\etaVec_n)
   \longrightarrow
   \Sigma^{-1}
\]
in operator norm. Since \(\Sigma\succeq0\),
\[
   \left|
      \tr\!\left(
         \Sigma\bigl(
            \E\Jmat(S_n\mid\etaVec_n)-\Sigma^{-1}
         \bigr)
      \right)
   \right|
   \le
   \tr(\Sigma)
   \left\|
      \E\Jmat(S_n\mid\etaVec_n)-\Sigma^{-1}
   \right\|_{\op}
   \longrightarrow0.
\]
Hence
\[
   \tr\!\left(
      \Sigma\,\E\Jmat(S_n\mid\etaVec_n)
   \right)
   \longrightarrow
   \tr(\Sigma\Sigma^{-1})
   =d.
\]
Together with \eqref{eq:relative-fisher-expansion}, this proves
\eqref{eq:relative-fisher-limit}.

For \(P_{\etaVec_n}\)-almost every realization of \(\etaVec_n\), the
Gaussian logarithmic Sobolev inequality \cite{gross1975logsobolev}  with reference measure
\(\gamma_\Sigma=N(0,\Sigma)\) gives
\[
   D\!\left(
      P_{S_n\mid\etaVec_n}\,\middle\|\,\gamma_\Sigma
   \right)
   \le
   \frac12
   \int_{\R^d}
   \left\|
      \Sigma^{1/2}
      \bigl(
         \rho_{S_n\mid\etaVec_n}(x)
         +\Sigma^{-1}x
      \bigr)
   \right\|^2
   f_{S_n\mid\etaVec_n}(x\mid\etaVec_n)\,dx.
\]
Averaging over \(\etaVec_n\) gives
\begin{equation}
   0
   \le
   \mathcal D_\Sigma(S_n\mid\etaVec_n)
   \le
   \frac12
   \mathcal J_\Sigma(S_n\mid\etaVec_n),
   \label{eq:relative-entropy-fisher-bound-proof}
\end{equation}
which is \eqref{eq:relative-entropy-fisher-bound}.  Together with
\eqref{eq:relative-fisher-limit}, this immediately yields
\[
   \mathcal D_\Sigma(S_n\mid\etaVec_n)\longrightarrow0,
\]
proving \eqref{eq:relative-entropy-limit}.

Finally, expanding the relative entropy with respect to
\(\gamma_\Sigma\) gives
\begin{align}
   \mathcal D_\Sigma(S_n\mid\etaVec_n)
   &=
   -\E h(S_n\mid\etaVec_n)
   +\frac{d}{2}\log(2\pi)
   +\frac12\log\det\Sigma
   \notag\\
   &\quad
   +\frac12
   \E\!\left[S_n^{\mathsf T}\Sigma^{-1}S_n\right].
\end{align}
Since
\[
   \E[S_n\mid\etaVec_n]=0,
   \qquad
   \E\Cov(S_n\mid\etaVec_n)=\Sigma,
\]
we have
\[
   \E\!\left[
      S_n^{\mathsf T}\Sigma^{-1}S_n
   \right]
   =
   \tr(\Sigma^{-1}\Sigma)
   =d.
\]
Hence
\begin{equation}
   \mathcal D_\Sigma(S_n\mid\etaVec_n)
   =
   \frac12
   \log\!\left((2\pi e)^d\det\Sigma\right)
   -
   \E h(S_n\mid\etaVec_n).
   \label{eq:relative-entropy-entropy-identity}
\end{equation}
Combining this identity with
\(\mathcal D_\Sigma(S_n\mid\etaVec_n)\to0\)
proves \eqref{eq:conditional-entropy-limit-from-fisher}.
\end{proof}

\subsection{Finite-constellation signaling: the low-SNR information slope}
\label{sec:low-snr-application}
The conditional channel model has a natural interpretation in terms
of receiver side information. Let
\[
   \widetilde S_n
   \triangleq
   \frac1{\sqrt n}\sum_{i=1}^n\xi_i
\]
denote the normalized aggregate before conditional centering. The
state \(\eta_i\) may describe receiver-observed information about the
\(i\)-th noise or interference component, such as its noise regime or
an effective noise law after receiver-side equalization. Given
\(\etaVec_n\), the conditional mean
\[
   m_n(\etaVec_n)
   \triangleq
   \frac1{\sqrt n}
   \sum_{i=1}^n\E[\xi_i\mid\eta_i]
\]
is known at the receiver, and
\[
   \widetilde S_n
   =
   m_n(\etaVec_n)+S_n.
\]

For any random vector \(U\) independent of the aggregate noise and
the state, consider the transmitted signal \(\sqrt\rho\,U\), where
\(\rho\ge0\) scales the signal power while the distribution of \(U\)
is held fixed. Then
\[
\begin{aligned}
&I\!\left(
   U;
   \bigl(\sqrt\rho\,U+\widetilde S_n,\etaVec_n\bigr)
\right)\\
&\quad=
I\!\left(
   U;
   \sqrt\rho\,U+\widetilde S_n
   \,\middle|\,
   \etaVec_n
\right)\\
&\quad=
I\!\left(
   U;
   \sqrt\rho\,U+S_n
   \,\middle|\,
   \etaVec_n
\right).
\end{aligned}
\]
The first equality follows from \(U\indep\etaVec_n\), while the second
follows by subtracting the receiver-known offset
\(m_n(\etaVec_n)\). Thus conditioning on \(\etaVec_n\) models the availability of the
state information at the receiver, while subtracting the corresponding
conditional mean removes only a receiver-known offset and therefore
does not change the mutual information.

Moreover,
\[
\begin{aligned}
   \E\Cov(S_n\mid\etaVec_n)
   &=
   \frac1n\sum_{i=1}^n
   \E\Cov(\xi_i\mid\eta_i)\\
   &=
   \E\Cov(\xi\mid\eta)
   =
   \Sigma
\end{aligned}
\]
for every \(n\). 
Hence the normalization by \(n^{-1/2}\) keeps the
averaged conditional noise covariance fixed. As \(n\) grows, \(S_n\) contains more independent components, each
scaled by \(n^{-1/2}\), while its averaged conditional covariance
remains equal to \(\Sigma\). The conditional Fisher-information CLT
describes the resulting Gaussianization in an averaged conditional
sense.

To analyze the low-SNR behavior of this channel, we use the classical
connection between weak-input mutual information and Fisher
information; see, e.g.,
\cite{prelov1993weak,rioul2011information}. In the present setting,
however, we require an averaged conditional expansion under the weak
Sobolev regularity used throughout this paper. We therefore record the
following finite-constellation expansion, whose proof is given in
Appendix~\ref{app:finite-constellation-weak-input}.

\begin{lemma}[Conditional finite-constellation weak-perturbation expansion]
\label{lem:conditional-weak-perturbation}
Let $(X,V)$ admit a jointly measurable conditional density in the first
coordinate and assume
\begin{equation}
   \E I(X\mid V)<\infty.
   \label{eq:weak-perturbation-fisher-assumption}
\end{equation}
Let $U\in\R^d$ be independent of $(X,V)$ and have finite support
\[
   \operatorname{supp}(U)
   =
   \{u_1,\ldots,u_M\},
\]
with
\[
   \P(U=u_j)=\pi_j>0,
   \qquad
   j=1,\ldots,M.
\]
Assume
\begin{equation}
   \E U=\sum_{j=1}^M\pi_j u_j=0,
   \label{eq:finite-constellation-centered}
\end{equation}
and write
\begin{equation}
   Q
   \triangleq
   \Cov(U)
   =
   \sum_{j=1}^M
   \pi_j u_j u_j^{\mathsf T}.
   \label{eq:finite-constellation-covariance}
\end{equation}
Then
\begin{equation}
   I(U;X+\sqrt\rho\,U\mid V)
   =
   \frac{\rho}{2}
   \tr\!\left(
      Q\,\E\Jmat(X\mid V)
   \right)
   +
   o(\rho),
   \qquad
   \rho\downarrow0.
   \label{eq:general-conditional-weak-expansion}
\end{equation}
Equivalently,
\begin{equation}
   \left.
   \frac{d}{d\rho}
   I(U;X+\sqrt\rho\,U\mid V)
   \right|_{\rho=0+}
   =
   \frac12
   \tr\!\left(
      Q\,\E\Jmat(X\mid V)
   \right).
   \label{eq:general-conditional-weak-derivative}
\end{equation}
\end{lemma}

\begin{remark}[Relation to classical weak-input expansions and the role of finite constellations]
\label{rem:relation-classical-weak-input}
The unconditional first-order coefficient in
\eqref{eq:general-conditional-weak-expansion} is classical.
Prelov and van der Meulen established weak-input asymptotics for a
broad class of vector memoryless channels
\cite{prelov1993weak}, while Rioul's generalized de Bruijn identity
gives the corresponding additive expansion for arbitrary
finite-covariance perturbations, including discrete ones
\cite[Sec.~II-D, Proposition~7, and (37)--(43)]
{rioul2011information}. Thus the role of
Lemma~\ref{lem:conditional-weak-perturbation} is not to provide a new
coefficient or to allow discrete inputs, but to establish an averaged
conditional expansion under the weak-Sobolev regularity required here.

Indeed, Rioul's derivation uses standing smoothness and decay
assumptions and a local parametric Kullback--Leibler expansion
\cite[Sec.~II-A]{rioul2011information}. In contrast, the present lemma
assumes only \(\E I(X\mid V)<\infty\) in the weak Sobolev sense, so the
conditional densities may be nonsmooth, may vanish, or may have compact
support. Moreover, a pointwise conditional remainder \(o_v(\rho)\)
cannot in general be averaged over \(V\).

For a finite constellation,
\[
   \bar p_t\ge\pi_jp_{j,t},
   \qquad
   \frac{p_{j,t}}{\bar p_t}\le\frac1{\pi_j},
\]
and the alphabet contains only finitely many components. These facts
provide the domination and finite summation needed for the
product-space KL--Hellinger argument. The restriction is also
operationally natural, covering fixed PSK- and QAM-type alphabets,
including probabilistically shaped finite constellations.
\end{remark}

\begin{remark}
\label{rem:centering-finite-constellation-input}
The centering assumption \(\E U=0\) in
Lemma~\ref{lem:conditional-weak-perturbation} is imposed only for
notational convenience in the proof.  Indeed, for a general
finite-constellation input \(U\), let
\[
   m_U=\E U,
   \qquad
   \widetilde U=U-m_U.
\]
Then
\[
   \E\widetilde U=0,
   \qquad
   \Cov(\widetilde U)=\Cov(U).
\]
Moreover, since deterministic translations of the channel output do not
change conditional mutual information,
\[
   I(U;X+\sqrt\rho\,U\mid V)
   =
   I\!\left(
      \widetilde U;
      X+\sqrt\rho\,\widetilde U
      \mid V
   \right).
\]
Hence the same expansion remains valid without the assumption
\(\E U=0\), with \(Q=\Cov(U)\).
\end{remark}

We now apply Lemma~\ref{lem:conditional-weak-perturbation} to the aggregate
noise.  Let $U\in\R^d$ be independent of
$\{(\xi_i,\eta_i)\}_{i\ge1}$ and have a fixed finite support.  Assume
\begin{equation}
   \E U=0,
   \qquad
   \Cov(U)=Q\succeq0.
   \label{eq:input-covariance-Q}
\end{equation}
For $\rho\ge0$, consider the additive channel with receiver side information
\begin{equation}
   Y_{n,\rho}
   =
   \sqrt\rho\,U+S_n,
   \qquad
   R_{n,U}(\rho)
   \triangleq
   I(U;Y_{n,\rho}\mid\etaVec_n).
   \label{eq:weak-signal-channel}
\end{equation}

\begin{proposition}[Conditional finite-constellation weak-signal identity]
\label{prop:conditional-low-snr}
Fix $n$ such that
\begin{equation}
   \E I(S_n\mid\etaVec_n)<\infty.
   \label{eq:low-snr-finite-fisher}
\end{equation}
Then every finite-constellation input $U$ satisfying
\eqref{eq:input-covariance-Q} obeys
\begin{equation}
   R_{n,U}(\rho)
   =
   \frac{\rho}{2}
   \tr\!\left(
      Q\,\E\Jmat(S_n\mid\etaVec_n)
   \right)
   +
   o_{n,U}(\rho),
   \qquad
   \rho\downarrow0,
   \label{eq:conditional-low-snr-expansion}
\end{equation}
where
\[
   \frac{o_{n,U}(\rho)}{\rho}
   \longrightarrow0
   \qquad
   \text{for each fixed $n$ and fixed finite-constellation input $U$.}
\]
Equivalently,
\begin{equation}
   R_{n,U}'(0+)
   =
   \frac12
   \tr\!\left(
      Q\,\E\Jmat(S_n\mid\etaVec_n)
   \right).
   \label{eq:conditional-low-snr-derivative}
\end{equation}
\end{proposition}

\begin{proof}
Apply Lemma~\ref{lem:conditional-weak-perturbation} with
\[
   (X,V)=(S_n,\etaVec_n).
\]
\end{proof}

\begin{corollary}[Universality of the asymptotic low-SNR slope]
\label{cor:low-snr-universality}
Under the assumptions of Theorem~\ref{thm:main}, suppose that the
finite-level Fisher condition holds.  Then, for every fixed
finite-constellation input law satisfying
\eqref{eq:input-covariance-Q},
\begin{equation}
   \lim_{n\to\infty}
   R_{n,U}'(0+)
   =
   \frac12
   \tr(Q\Sigma^{-1}).
   \label{eq:low-snr-gaussian-benchmark}
\end{equation}
The right-hand side is exactly the low-SNR slope obtained when the
additive noise is Gaussian with covariance $\Sigma$.
\end{corollary}

\begin{proof}
By Proposition~\ref{prop:conditional-low-snr},
\[
   R_{n,U}'(0+)
   =
   \frac12
   \tr\!\left(
      Q\,\E\Jmat(S_n\mid\etaVec_n)
   \right).
\]
Corollary~\ref{cor:matrix} gives
\[
   \E\Jmat(S_n\mid\etaVec_n)
   \longrightarrow
   \Sigma^{-1}
\]
in operator norm.  Since $Q$ is fixed,
\[
   \tr\!\left(
      Q\,\E\Jmat(S_n\mid\etaVec_n)
   \right)
   \longrightarrow
   \tr(Q\Sigma^{-1}),
\]
which proves the claim.
\end{proof}

\begin{remark}[Dependence on the input covariance]
\label{rem:finite-constellation-universality}

Within the class of fixed finite constellations,
Proposition~\ref{prop:conditional-low-snr} shows that the
first-order slope depends on the input law only through its covariance
\(Q=\Cov(U)\), while
Corollary~\ref{cor:low-snr-universality} identifies its limiting value
as
\[
   \frac12\tr(Q\Sigma^{-1}).
\]
Consequently, once \(Q\) is fixed, the detailed constellation geometry
and symbol probabilities have no further effect on the limiting
first-order slope. In particular, when \(d=1\),
\(\Sigma=\sigma^2\), and \(\operatorname{Var}(U)=1\), this slope is
\(1/(2\sigma^2)\). Hence BPSK and every other fixed zero-mean
unit-energy scalar finite constellation have the same limiting
first-order slope.
\end{remark}

For the corresponding optimization statement, let
\(\mathcal U_{\mathrm{fin}}(P)\) denote the class of all
\(\mathbb R^d\)-valued finite-support random vectors \(U\), independent
of \(\{(\xi_i,\eta_i)\}_{i\ge1}\), such that
\[
   \E U=0,
   \qquad
   \E\|U\|^2\le P.
\]
Here \(\E\|U\|^2\) is the average input energy per channel use.
Although it is automatically finite for a finite constellation, the
bound \(P\) imposes a common energy budget and makes the slope
optimization well posed, since replacing \(U\) by \(aU\) multiplies
the first-order slope by \(a^2\). As in the classical weak-input result
of Prelov and van der Meulen \cite{prelov1993weak}, the optimization
below has a largest-eigenvalue form. In the present conditional setting,
the governing matrix is the averaged conditional Fisher information
matrix, whose convergence to \(\Sigma^{-1}\) follows from
Corollary~\ref{cor:matrix}.

\begin{corollary}[Optimal first-order slope among fixed finite constellations]
\label{cor:optimal-weak-slope}
Under the assumptions of Theorem~\ref{thm:main}, suppose that the
finite-level Fisher condition holds. Let
\[
   A_n
   \triangleq
   \E\Jmat(S_n\mid\etaVec_n),
\]
and fix \(P>0\). For every sufficiently large \(n\),
\begin{equation}
   \sup_{U\in\mathcal U_{\mathrm{fin}}(P)}
   R_{n,U}'(0+)
   =
   \frac P2\lambda_{\max}(A_n).
   \label{eq:optimal-first-order-slope}
\end{equation}
Moreover, the first-order slopes converge uniformly over the
energy-constrained finite-con\-stel\-lation class:
\begin{equation}
\begin{aligned}
   &\sup_{U\in\mathcal U_{\mathrm{fin}}(P)}
   \left|
      R_{n,U}'(0+)
      -
      \frac12
      \tr\!\left(
         \Cov(U)\Sigma^{-1}
      \right)
   \right|
   \\
   &\qquad\le
   \frac P2
   \left\|
      A_n-\Sigma^{-1}
   \right\|_{\op}
   \longrightarrow0.
\end{aligned}
\label{eq:uniform-first-order-slope}
\end{equation}
Consequently,
\begin{equation}
   \lim_{n\to\infty}
   \sup_{U\in\mathcal U_{\mathrm{fin}}(P)}
   R_{n,U}'(0+)
   =
   \frac P2
   \lambda_{\max}(\Sigma^{-1}).
   \label{eq:optimal-first-order-slope-limit}
\end{equation}
\end{corollary}

\begin{proof}
For $U\in\mathcal U_{\mathrm{fin}}(P)$, let
\[
   Q=\Cov(U).
\]
Since $\E U=0$,
\[
   \tr Q
   =
   \E\|U\|^2
   \le P.
\]
Proposition~\ref{prop:conditional-low-snr} gives
\[
   R_{n,U}'(0+)
   =
   \frac12\tr(QA_n).
\]
Since $A_n\succeq0$,
\[
   \tr(QA_n)
   \le
   \lambda_{\max}(A_n)\tr Q
   \le
   P\lambda_{\max}(A_n).
\]
Hence
\[
   \sup_{U\in\mathcal U_{\mathrm{fin}}(P)}
   R_{n,U}'(0+)
   \le
   \frac P2\lambda_{\max}(A_n).
\]

Conversely, let $v_n$ be a unit eigenvector corresponding to
$\lambda_{\max}(A_n)$, and take the binary input
\[
   U_n=
   \begin{cases}
      \sqrt P\,v_n, & \text{with probability }1/2,\\
      -\sqrt P\,v_n, & \text{with probability }1/2.
   \end{cases}
\]
Then $U_n\in\mathcal U_{\mathrm{fin}}(P)$ and
\[
   \Cov(U_n)
   =
   P\,v_nv_n^{\mathsf T}.
\]
Therefore
\[
   R_{n,U_n}'(0+)
   =
   \frac P2
   v_n^{\mathsf T}A_nv_n
   =
   \frac P2
   \lambda_{\max}(A_n).
\]
This proves \eqref{eq:optimal-first-order-slope}.

Finally, Corollary~\ref{cor:matrix} gives
\[
   A_n\longrightarrow\Sigma^{-1}
\]
in operator norm. For every
\(U\in\mathcal U_{\mathrm{fin}}(P)\), writing
\(Q=\Cov(U)\), we have
\begin{align}
   &
   \left|
      R_{n,U}'(0+)
      -
      \frac12\tr(Q\Sigma^{-1})
   \right|
   \notag\\
   &\qquad=
   \frac12
   \left|
      \tr\!\left(
         Q(A_n-\Sigma^{-1})
      \right)
   \right|
   \notag\\
   &\qquad\le
   \frac12
   \tr(Q)
   \|A_n-\Sigma^{-1}\|_{\op}
   \notag\\
   &\qquad\le
   \frac P2
   \|A_n-\Sigma^{-1}\|_{\op}.
\end{align}
Taking the supremum over
\(U\in\mathcal U_{\mathrm{fin}}(P)\) proves
\eqref{eq:uniform-first-order-slope}.

Moreover,
\[
   \lambda_{\max}(A_n)
   \longrightarrow
   \lambda_{\max}(\Sigma^{-1}),
\]
and hence \eqref{eq:optimal-first-order-slope} yields
\eqref{eq:optimal-first-order-slope-limit}.
\end{proof}

\subsection{Gaussian signaling at fixed signal covariance}
\label{sec:gaussian-signal-application}
In the same receiver-side-information model, 
consider Gaussian signaling at a fixed input covariance. 
The following bound is consistent with the classical worst-additive-noise 
principle \cite{diggavi2001worst}. Let $G_Q\sim N(0,Q)$ be independent of
$(S_n,\etaVec_n)$ and define
\begin{align}
   R_n^{\rm G}(Q)&=I(G_Q;G_Q+S_n\mid\etaVec_n),\\
   R_{\rm AWGN}(Q)&=\frac12\log\det\!\left(\Id+\Sigma^{-1/2}Q\Sigma^{-1/2}\right).
   \label{eq:gaussian-signaling-benchmark}
\end{align}

\begin{proposition}[Fisher control of the Gaussian-signaling gap]
\label{prop:gaussian-signaling-gap}
Under the assumptions of Corollary~\ref{cor:relative-information}, for every
$Q\succeq0$,
\begin{equation}
   0\le R_n^{\rm G}(Q)-R_{\rm AWGN}(Q)
   \le\mathcal D_\Sigma(S_n\mid\etaVec_n)
   \le\frac12\mathcal J_\Sigma(S_n\mid\etaVec_n).
   \label{eq:gaussian-signaling-gap-bound}
\end{equation}
Consequently,
\begin{equation}
   R_n^{\rm G}(Q)\longrightarrow R_{\rm AWGN}(Q),
   \label{eq:gaussian-signaling-gap-limit}
\end{equation}
and, in fact,
\begin{equation}
   \sup_{Q\succeq0}\bigl(R_n^{\rm G}(Q)-R_{\rm AWGN}(Q)\bigr)
   \longrightarrow0.
\end{equation}
\end{proposition}

\begin{proof}
Set
\begin{align*}
   \Delta_n&=\mathcal D_\Sigma(S_n\mid\etaVec_n),\\
   \Delta_{n,Q}^{\rm out}
   &=\E D\!\left(P_{S_n+G_Q\mid\etaVec_n}\middle\|N(0,\Sigma+Q)\right).
\end{align*}
Because the conditional means vanish and the averaged conditional
covariances are $\Sigma$ and $\Sigma+Q$,
\begin{align*}
   \Delta_n
   &=h(N(0,\Sigma))-\E h(S_n\mid\etaVec_n),\\
   \Delta_{n,Q}^{\rm out}
   &=h(N(0,\Sigma+Q))-\E h(S_n+G_Q\mid\etaVec_n).
\end{align*}
Subtracting gives the exact identity
\begin{equation}
   R_n^{\rm G}(Q)-R_{\rm AWGN}(Q)
   =\Delta_n-\Delta_{n,Q}^{\rm out}.
   \label{eq:exact-gaussian-gap-identity}
\end{equation}
Let \(\gamma_C\) denote the Gaussian probability measure \(N(0,C)\)
for \(C\succeq0\). For \(P_{\etaVec_n}\)-almost every \(v\),
\[
\begin{aligned}
&D\!\left(
   P_{S_n+G_Q\mid\etaVec_n=v}
   \middle\|
   \gamma_{\Sigma+Q}
\right)\\
&\quad=
D\!\left(
   P_{S_n\mid\etaVec_n=v}*\gamma_Q
   \middle\|
   \gamma_\Sigma*\gamma_Q
\right)\\
&\quad\le
D\!\left(
   P_{S_n\mid\etaVec_n=v}
   \middle\|
   \gamma_\Sigma
\right).
\end{aligned}
\]
Here the equality follows from the independence of \(G_Q\) and the
Gaussian convolution identity
\(\gamma_\Sigma*\gamma_Q=\gamma_{\Sigma+Q}\), while the inequality
follows from the data-processing inequality for relative entropy. Averaging over \(\etaVec_n\), and using the nonnegativity of relative
entropy, gives
\[
   0\le\Delta_{n,Q}^{\rm out}\le\Delta_n.
\]
The final bound in \eqref{eq:gaussian-signaling-gap-bound} follows from
\eqref{eq:relative-entropy-fisher-bound}. Since the resulting upper
bound is independent of \(Q\), the uniform conclusion follows.

\end{proof}

\appendices

\section{The multivariate log-concave second-order Fisher inequality}
\label{app:logconcave-K}

Throughout this appendix, the dimension $d\geq1$ is fixed. We first work
with strictly positive log-concave probability densities
$f$ on $\mathbb R^d$ satisfying
\[
    f\in C^2(\mathbb R^d)
\]
and, for every multi-index $\alpha$ with $|\alpha|\leq2$,
\begin{equation}
    D^\alpha f
    \in
    L^1(\mathbb R^d)\cap C_0(\mathbb R^d),
    \label{eq:logconcave-boundary-decay}
\end{equation}
where $C_0(\mathbb R^d)$ denotes the space of continuous functions
vanishing at infinity. These assumptions are used only to justify
differentiation under convolution and the integrations by parts below.
In particular, no derivatives of order greater than two are required.

For such a density $f$, define the second-order score matrix by
\begin{equation}
    \Psi_f(x)
    \triangleq
    -\nabla^2\log f(x).
    \label{eq:second-order-score-matrix}
\end{equation}
Since $f$ is log-concave,
\[
    \Psi_f(x)\succeq0,
    \qquad
    x\in\mathbb R^d.
\]
In this notation,
\begin{equation}
    K(f)
    =
    \int_{\mathbb R^d}
    \|\Psi_f(x)\|_{\HS}^2 f(x)\,dx.
    \label{eq:K-second-order-score}
\end{equation}

We now derive the convolution inequality used in the main argument.

\begin{lemma}
\label{lem:4.2}
Let $X$ and $Y$ be independent random vectors in $\mathbb R^d$ with
log-concave densities satisfying the regularity assumptions above, and assume
that
\[
    K(X)<\infty,
    \qquad
    K(Y)<\infty.
\]
Then, for every $\lambda\in[0,1]$,
\begin{equation}
    K\left(
        \lambda X+\sqrt{1-\lambda^2}\,Y
    \right)
    \leq
    \lambda^2 K(X)
    +
    (1-\lambda^2)K(Y).
    \label{eq:second-order-fisher-convolution}
\end{equation}
\end{lemma}

\begin{proof}
   Toscani's multivariate convolution inequality
\cite[Theorem~6]{toscani2015strengthened}, rewritten in our notation,
states that for independent log-concave random vectors $U,V\in\mathbb R^d$
and arbitrary $a,b>0$,
\begin{align}
    K(U+V)
  &  \leq
    \frac{a^4}{(a+b)^4}K(U)
    +
    \frac{b^4}{(a+b)^4}K(V)
    +
    \frac{2a^2b^2}{(a+b)^4}
    \mathcal H(U,V),\notag \\ 
    &\leq \frac{a^4}{(a+b)^4}K(U)
    +
    \frac{b^4}{(a+b)^4}K(V)
    +
    \frac{2a^2b^2}{(a+b)^4} \sqrt{K(U)}\sqrt{K(V)}
    \label{eq:toscani-second-order}
\end{align}
where 
\begin{equation}
   \mathcal H(f,g)
    \triangleq \sum_{i,j=1}^d \int_{\mathbb R^d}
    \frac{\partial_i f(x)\partial_j f(x)}{f(x)}
    \,dx \int_{\mathbb R^d}
    \frac{\partial_i g(x)\partial_j g(x)}{g(x)}
    \,dx.
\end{equation}

Fix $\lambda\in(0,1)$ and put
\[
    \mu
    \triangleq
    \sqrt{1-\lambda^2}.
\]
Apply \eqref{eq:toscani-second-order} with
\[
    U=\lambda X,
    \qquad
    V=\mu Y,
\]
and
\[
    a=\lambda^2,
    \qquad
    b=\mu^2=1-\lambda^2.
\]
Then $a+b=1$. Using the scaling identities
\[
    K(cX)=c^{-4}K(X),
\]
Therefore,
\begin{align}
    K\left(
        \lambda X+\mu Y
    \right)
    &\leq
    \lambda^4 K(X)
    +
    \mu^4 K(Y)
    +
    2\lambda^2\mu^2\sqrt{K(X)K(Y)}\\
    &\leq  \lambda^4 K(X)+
    \mu^4 K(Y)+
   \lambda^2\mu^2(K(X)+K(Y)).
    \label{eq:intermediate-second-order-bound}
\end{align}

Substituting this into
\eqref{eq:intermediate-second-order-bound} yields
\begin{align}
    K\left(
        \lambda X+\mu Y
    \right)
    &\leq
    \bigl(\lambda^4+\lambda^2\mu^2\bigr)K(X)
    +
    \bigl(\mu^4+\lambda^2\mu^2\bigr)K(Y)
    \notag\\
    &=
    \lambda^2(\lambda^2+\mu^2)K(X)
    +
    \mu^2(\lambda^2+\mu^2)K(Y)
    \notag\\
    &=
    \lambda^2K(X)
    +
    \mu^2K(Y)
    \notag\\
    &=
    \lambda^2K(X)
    +
    (1-\lambda^2)K(Y).
\end{align}
This proves \eqref{eq:second-order-fisher-convolution} for
$\lambda\in(0,1)$. The endpoint cases $\lambda=0$ and $\lambda=1$ are
immediate.
\end{proof}

\begin{remark}[Gaussian-regularized log-concave densities]
\label{rem:Gaussian-regularized-K}
Let $p$ be a possibly nonsmooth log-concave density and
$p_\varepsilon=p*\gamma_\varepsilon$, where $\gamma_\varepsilon$ is the
$N(0,\varepsilon\Id)$ density. Gaussian convolution makes
$p_\varepsilon$ strictly positive and smooth, preserves log-concavity, and
satisfies
\[
   D^\alpha p_\varepsilon=p*D^\alpha\gamma_\varepsilon,
   \qquad |\alpha|\le2.
\]
Since $D^\alpha\gamma_\varepsilon\in L^1\cap C_0$, the density
$p_\varepsilon$ satisfies \eqref{eq:logconcave-boundary-decay}. Thus
Lemma~\ref{lem:4.2} applies directly to all Gaussian-smoothed conditional laws
used in the paper.
\end{remark}

\section{An alternative one-dimensional proof}
\label{app:scalar-alternative}

Corollary~\ref{cor:scalar} already follows from Theorem~\ref{thm:main} by
setting $d=1$. This appendix gives an independent proof following the
conditional entropic route of \cite{ma2024entropic} and the first-order
Gaussian-smoothed continuity method of \cite{ye2026finiteentropy}. The new
technical input is the continuity of the second-order Fisher production $K$
proved in Appendix~\ref{app:smoothed-continuity}.

Write
\begin{equation}
   I_n=\E I(S_n\mid\etaVec_n),
   \qquad n\ge1.
   \label{eq:alternative-scalar-In}
\end{equation}
Necessity is immediate. For sufficiency, fix $n_0$ with $I_{n_0}<\infty$ and
put $N_m=2^mn_0$. For $t\ge0$, let
\[
   S_{N_m,t}=S_{N_m}+\sqrt t\,G,
\]
where $G\sim N(0,1)$ is independent, and define
\begin{equation}
   I_m(t)=\E I(S_{N_m,t}\mid\etaVec_{N_m}),
   \qquad
   K_m(t)=\E K(S_{N_m,t}\mid\etaVec_{N_m}).
   \label{eq:alternative-smoothed-IK}
\end{equation}
The Fisher dissipation identity gives
\begin{equation}
   I_m(0)-I_m(t)=\int_0^tK_m(s)\,ds.
   \label{eq:alternative-integrated-dissipation}
\end{equation}
Splitting the $N_{m+1}$ summands into two independent blocks and applying
Lemma~\ref{lem:4.2} conditionally yields
\begin{equation}
   K_{m+1}(t)\le K_m(t),
   \qquad t>0.
   \label{eq:alternative-K-monotone}
\end{equation}
By \eqref{eq:weighted-fisher}, we have 
$I_{m+1}(t)\le I_m(t)$.

\begin{proposition}
\label{prop:alternative-regularized-limits}
For every fixed $t>0$,
\begin{equation}
   \lim_{m\rightarrow \infty}K_m(t)=\frac1{(\sigma^2+t)^2}.
   \label{eq:alternative-regularized-limits}
\end{equation}
\end{proposition}

\begin{proof}
For $n\ge1$, define
\[
   \widetilde\xi_i
   =
   \xi_i+\sqrt t\,G_i,
   \qquad
   \overline W_{n,t}
   =
   \frac1{\sqrt n}\sum_{i=1}^n
   \bigl(
      \widetilde\xi_i
      -
      \E[\widetilde\xi_i\mid\eta_i]
   \bigr),
\]
where the $G_i$ are independent standard Gaussian random variables,
independent of the original pairs. Conditionally on $\etaVec_n$,
$\overline W_{n,t}$ has the same law as $S_n+\sqrt t\,G$.

By the argument in the appendix of~\cite{ma2024entropic}, the family
\[
   \left\{
      \overline W_{n,t}^{\,2}
   \right\}_{n\ge1}
\]
is uniformly integrable. Define
\begin{equation}
   \ell(R)
   =
   \sup_{n\ge1}
   \mathbb E
   \left[
      \overline W_{n,t}^{\,2}
      \mathbf 1_{\{
         |\overline W_{n,t}|\ge R
      \}}
   \right].
\end{equation}
Then $\ell$ is nonnegative and decreasing, and
\[
   \ell(R)\longrightarrow0
   \qquad
   \text{as }R\to\infty.
\]
Therefore,
\begin{equation}
   \sup_{m\ge0}
   \E
   \tau_{S_{N_m,t}\mid\etaVec_{N_m}}(R)
   \le
   \ell(R),
   \qquad R\ge0.
   \label{eq:alternative-tail-control}
\end{equation}

Moreover,
\begin{equation}
   \E
   \operatorname{Var}
   (S_{N_m,t}\mid\etaVec_{N_m})
   =
   \sigma^2+t.
   \label{eq:alternative-average-variance}
\end{equation}
Thus, by~\cite{ye2026finiteentropy},
\begin{align}
\lim_{m\to\infty}
\mathbb E\,
I\left(
   S_{N_m,t}
   \mid
   \etaVec_{N_m}
\right)
&=
\lim_{m\to\infty}
\mathbb E\,
I\left(
   S_{N_m}+\sqrt t\,G
   \mid
   \etaVec_{N_m}
\right)
\notag\\
&=
\frac{1}{\sigma^2+t}.
\label{eq:alternative-fisher-limit}
\end{align}

By~\eqref{eq:K-bound},
\begin{equation}
   \mathbb E\,
   K\left(
      S_{N_m,t}
      \mid
      \etaVec_{N_m}
   \right)
   \le
   \frac{1}{t^2}.
\end{equation}
Moreover, it follows from~\cite{villani2006short} that
\begin{equation}
   I(X)^2
   \le
   K(X).
   \label{eq:K-I-comparison}
\end{equation}
Therefore,
\begin{align}
\mathbb E\,
I\left(
   S_{N_m,t}
   \mid
   \etaVec_{N_m}
\right)
&\le
\sqrt{
   \mathbb E
   \left[
      I\left(
         S_{N_m,t}
         \mid
         \etaVec_{N_m}
      \right)^2
   \right]
}
\notag\\
&\le
\sqrt{
   \mathbb E\,
   K\left(
      S_{N_m,t}
      \mid
      \etaVec_{N_m}
   \right)
}
\notag\\
&\le
\frac{1}{t}.
\end{align}

Since
\begin{align}
\sup_{m\ge0}
\mathbb E\,
I\left(
   S_{N_m,t}
   \mid
   \etaVec_{N_m}
\right)
&\le
\frac{1}{t},
\\
\sup_{m\ge0}
\mathbb E\,
\tau_{S_{N_m,t}\mid\etaVec_{N_m}}(R)
&\le
\ell(R),
\end{align}
it follows from~\cite{ma2024entropic} that
\begin{equation}
   \mathfrak D
   (S_{N_m,t}\mid\etaVec_{N_m})
   \overset{P}{\longrightarrow}
   0,
   \qquad
   \operatorname{Var}
   (S_{N_m,t}\mid\etaVec_{N_m})
   \overset{P}{\longrightarrow}
   \sigma^2+t.
   \label{eq:alternative-entropic-variance-convergence}
\end{equation}
Here,
\begin{equation}
   \mathfrak D(X\mid Y=y)
   =
   D\!\left(
      P_{X\mid Y=y}
      \middle\|
      N\!\left(
         0,
         \operatorname{Var}(X\mid Y=y)
      \right)
   \right).
   \label{eq:matching-gaussian-deficit}
\end{equation}

As in the proof of Proposition~5.16 in~\cite{ma2024entropic},
for every $\eta>0$, there exists a nonnegative decreasing function
$\ell_1$ satisfying
\[
   \ell_1(R)\longrightarrow0
   \qquad
   \text{as }R\to\infty,
\]
such that
\begin{equation}
   \P\!\left(
      \tau_{S_{N_m,t}\mid\etaVec_{N_m}}(R)
      \le
      \ell_1(R)
      \text{ for all }R\ge0
   \right)
   \ge
   1-\eta.
   \label{eq:alternative-high-prob-tail}
\end{equation}

Let $\mathcal R_m$ denote the event in
\eqref{eq:alternative-high-prob-tail}. For $\delta>0$ and
$0<\varepsilon<\sigma^2+t$, define
\[
   \mathcal H_m
   =
   \left\{
      \mathfrak D
      (S_{N_m,t}\mid\etaVec_{N_m})
      \le
      \delta
   \right\},
\]
and
\[
   \mathcal F_m
   =
   \left\{
      \left|
         \operatorname{Var}
         (S_{N_m,t}\mid\etaVec_{N_m})
         -
         (\sigma^2+t)
      \right|
      \le
      \varepsilon
   \right\}.
\]
Set
\[
   \mathcal W_m
   =
   \mathcal R_m
   \cap
   \mathcal H_m
   \cap
   \mathcal F_m.
\]

By Pinsker's inequality~\cite{pinsker1964information} and
Appendix~\ref{app:smoothed-continuity}, $\delta$ can be chosen
sufficiently small so that, on the event $\mathcal W_m$,
\begin{equation}
\left|
K\left(
   S_{N_m,t}
   \mid
   \etaVec_{N_m}
\right)
-
\frac{1}{
   \operatorname{Var}
   \left(
      S_{N_m,t}
      \mid
      \etaVec_{N_m}
   \right)^2
}
\right|
\le
\varepsilon.
\end{equation}

By~\eqref{eq:alternative-entropic-variance-convergence} and
\eqref{eq:alternative-high-prob-tail},
\[
   \liminf_{m\to\infty}
   \P(\mathcal W_m)
   \ge
   1-3\eta.
\]
Therefore,
\begin{align}
K_m(t)
&=
\mathbb E\,
K\left(
   S_{N_m,t}
   \mid
   \etaVec_{N_m}
\right)
\notag\\
&=
\mathbb E
\left[
   K\left(
      S_{N_m,t}
      \mid
      \etaVec_{N_m}
   \right)
   \mathbf 1_{\mathcal W_m}
\right]
+
\mathbb E
\left[
   K\left(
      S_{N_m,t}
      \mid
      \etaVec_{N_m}
   \right)
   \mathbf 1_{\mathcal W_m^c}
\right]
\notag\\
&\le
\mathbb E
\left[
   \left(
      \varepsilon
      +
      \frac{1}{
         \operatorname{Var}
         \left(
            S_{N_m,t}
            \mid
            \etaVec_{N_m}
         \right)^2
      }
   \right)
   \mathbf 1_{\mathcal W_m}
\right]
+
\frac{1}{t^2}
\P(\mathcal W_m^c)
\notag\\
&\le
\varepsilon
+
\frac{1}{(\sigma^2+t-\varepsilon)^2}
+
\frac{1}{t^2}
\P(\mathcal W_m^c).
\end{align}
Hence,
\begin{align}
\limsup_{m\to\infty} K_m(t)
&\le
\varepsilon
+
\frac{1}{(\sigma^2+t-\varepsilon)^2}
+
\frac{1}{t^2}
\limsup_{m\to\infty}
\P(\mathcal W_m^c)
\notag\\
&\le
\varepsilon
+
\frac{1}{(\sigma^2+t-\varepsilon)^2}
+
\frac{3\eta}{t^2}.
\end{align}
Letting $\varepsilon,\eta\downarrow0$ yields
\begin{equation}
   \limsup_{m\to\infty}
   K_m(t)
   \le
   \frac{1}{(\sigma^2+t)^2}.
\end{equation}

On the other hand, by~\eqref{eq:K-I-comparison},
\begin{align}
K_m(t)
&=
\mathbb E\,
K\left(
   S_{N_m,t}
   \mid
   \etaVec_{N_m}
\right)
\notag\\
&\ge
\mathbb E
\left[
   I\left(
      S_{N_m,t}
      \mid
      \etaVec_{N_m}
   \right)^2
\right]
\notag\\
&\ge
\left(
   \mathbb E\,
   I\left(
      S_{N_m,t}
      \mid
      \etaVec_{N_m}
   \right)
\right)^2.
\end{align}
Together with~\eqref{eq:alternative-fisher-limit}, this gives
\begin{equation}
   \liminf_{m\to\infty}
   K_m(t)
   \ge
   \frac{1}{(\sigma^2+t)^2}.
\end{equation}
Consequently,
\begin{equation}
   \lim_{m\to\infty}
   K_m(t)
   =
   \frac{1}{(\sigma^2+t)^2}.
\end{equation}

\end{proof}

We now remove the smoothing. By
\eqref{eq:alternative-integrated-dissipation},
\[
   I_m(0)=I_m(t)+\int_0^tK_m(s)\,ds.
\]
By \eqref{eq:alternative-K-monotone},
$0\le K_m(s)\le K_0(s)$. Since $I_{n_0}<\infty$,
\[
   \int_0^tK_0(s)\,ds=I_0(0)-I_0(t)\le I_{n_0}<\infty.
\]
Proposition~\ref{prop:alternative-regularized-limits} and dominated
convergence therefore give
\begin{align*}
   \lim_{m\to\infty}I_{N_m}
   &=\frac1{\sigma^2+t}+\int_0^t\frac{ds}{(\sigma^2+s)^2}\\
   &=\frac1{\sigma^2}.
\end{align*}

Finally, the one-component Fisher inequality propagates finiteness from
$n_0$ to every $n\ge n_0$. Splitting $S_{m+n}$ into two conditionally
independent normalized blocks gives
\begin{equation}
   (m+n)I_{m+n}\le mI_m+nI_n.
   \label{eq:alternative-subadditivity}
\end{equation}
Thus $a_n=nI_n$ is an extended-real subadditive sequence that is eventually
finite. Fekete's lemma implies that $I_n=a_n/n$ has a limit, and the dyadic
subsequence identifies it as $1/\sigma^2$. This completes the independent
proof of Corollary~\ref{cor:scalar}.

\section{One-dimensional compactness and continuity after Gaussian smoothing}
\label{app:smoothed-continuity}

For $a>0$ and a nonnegative function $\ell$ with $\ell(R)\to0$, define
\begin{equation}
   \mathcal I_{a,\ell}
   =\left\{f_{X_0+G_a}:
      \tau_{X_0+G_a}(R)\le\ell(R)\ \forall R>0,
      \ \E X_0=0
   \right\},
   \label{eq:smoothed-tail-class}
\end{equation}
where $G_a\sim N(0,a)$ is independent of $X_0$. Let
$\overline{\mathcal I}_{a,\ell}$ denote the closure of this class in
$L^1_{1+x^2}=L^1(\R,(1+x^2)dx)$.

We first record the generalized dominated convergence principle used below.

\begin{lemma}[Generalized dominated convergence]
\label{lem:generalized-dct}
Let $F_j,G_j,F,G$ be nonnegative measurable functions such that
$F_j\le G_j$, $F_j\to F$ and $G_j\to G$ almost everywhere, and
$\int G_j\to\int G<\infty$. Then $F_j\to F$ in $L^1$.
\end{lemma}

\begin{proof}
Fatou's lemma applied to $F_j$ gives
$\int F\le\liminf_j\int F_j$. Applying it to $G_j-F_j$ and using
$\int G_j\to\int G$ gives
$\limsup_j\int F_j\le\int F$. Hence the integrals converge, and Scheff\'e's
lemma yields $L^1$ convergence.
\end{proof}

\begin{lemma}[Continuity of the second-order Fisher production]
\label{lem:4.1}
Let $p_n,p\in\overline{\mathcal I}_{a,\ell}$ and suppose
$\|p_n-p\|_{L^1}\to0$. Then
\begin{equation}
   K(p_n)\longrightarrow K(p).
   \label{eq:K-continuity}
\end{equation}
\end{lemma}

\begin{proof}
It suffices first to prove the following core statement: if
$q_j\in\mathcal I_{a,\ell}$, $q\in\overline{\mathcal I}_{a,\ell}$, and
$q_j\to q$ in $L^1$, then $K(q_j)\to K(q)$. We prove this statement and
then extend it to arbitrary sequences in the closure.

By deterministic rescaling, it suffices to consider $a\in(0,1)$. Every
$q_j=f_{X_{0,j}+G_a}$ can be represented as
\[
   X_{0,j}+G_a
   \overset{d}{=}
   e^{-t}Z_j+\sqrt{1-e^{-2t}}\,G,
   \qquad
   t=-\frac12\log(1-a/2),
\]
where
\[
   e^{-t}Z_j=X_{0,j}+\widehat G_{a/2},
   \qquad
   \widehat G_{a/2}\sim N(0,a/2).
\]
Let $f_j$ be the density of $e^{-t}Z_j$, and let $\phi_t$ be the
density of
\[
   N(0,1-e^{-2t})=N(0,a/2).
\]
Then
\[
   q_j=f_j\ast\phi_t.
\]

It follows from~\cite{ye2026finiteentropy} that every subsequence
admits a further subsequence, still denoted by $j$, and a density $f$
such that
\begin{equation}
   q_j\overset{L^1_{1+x^2}}{\longrightarrow} q,\qquad
   f_j\overset{L^1_{1+x^2}}{\longrightarrow} f,\qquad
   q_j'\overset{a.e.}{\longrightarrow} q',
   \label{eq:compact-subsequence-convergence}
\end{equation}
with $q=f\ast\phi_t$.

Differentiation under convolution gives, with
$\sigma_t^2=1-e^{-2t}$,
\[
   q_j''(y)
   =
   \int_{\R}
   f_j(x)
   \left(
      \frac{(y-x)^2}{\sigma_t^4}
      -
      \frac1{\sigma_t^2}
   \right)
   \phi_t(y-x)\,dx.
\]
Consequently,
\begin{equation}
   |q_j''(y)-q''(y)|
   \le
   (C_t+C_t'y^2)\|f_j-f\|_{L^1}
   +
   C_t'\|f_j-f\|_{L^1_{1+x^2}},
   \label{eq:second-derivative-convergence-bound}
\end{equation}
where
\[
   C_t
   =
   \frac{1}
   {(2\pi)^{1/2}(1-e^{-2t})^{3/2}},
   \qquad
   C_t'
   =
   \frac{2}
   {(1-e^{-2t})^{5/2}(2\pi)^{1/2}}.
\]
Hence,
\[
   q_j''(y)\longrightarrow q''(y)
   \qquad\text{for every }y.
\]
Moreover, since $f_j\to f$ in $L^1$ and
$q_j=f_j\ast\phi_t$, we have
\[
   q_j(y)\longrightarrow q(y)
   \qquad\text{for every }y.
\]
Since Gaussian convolution yields $q_j(y)>0$ and $q(y)>0$, it follows
that
\begin{equation}
   \left[
      \frac{q_j''}{q_j}
      -
      \left(\frac{q_j'}{q_j}\right)^2
   \right]^2q_j
   \longrightarrow
   \left[
      \frac{q''}{q}
      -
      \left(\frac{q'}{q}\right)^2
   \right]^2q
   \quad\text{a.e.}
   \label{eq:K-integrand-ae}
\end{equation}

It remains to establish uniform integrability of the integrands in
\eqref{eq:K-integrand-ae}. By H\"older's inequality,
\begin{equation}
   (q_j'(y))^4
   \le
   \frac{q_j(y)^3}{(1-e^{-2t})^4}
   \mathbb E\!\left[
      (y-e^{-t}Z_j)^4
      \phi_t(y-e^{-t}Z_j)
   \right].
\end{equation}
Therefore,
\begin{equation}
   \frac{(q_j'(y))^4}{q_j(y)^3}
   \le
   \frac{1}{(1-e^{-2t})^4}
   \mathbb E\!\left[
      (y-e^{-t}Z_j)^4
      \phi_t(y-e^{-t}Z_j)
   \right].
\end{equation}

Let
\begin{align}
   U_t(y)
   &=
   \frac{y^4\phi_t(y)}{\phi_{2t}(y)}
   \notag\\
   &=
   y^4\sqrt{1+e^{-2t}}
   \exp\left(
      -\frac{e^{-2t}y^2}{2(1-e^{-4t})}
   \right).
\end{align}
Note that $U_t(y)\ge0$ and
\begin{equation}
   \sup_{y\in\R}U_t(y)<\infty.
\end{equation}
Hence,
\begin{align}
   \frac{(q_j'(y))^4}{q_j(y)^3}
   &\le
   \frac{\sup_{y\in\R}U_t(y)}
        {(1-e^{-2t})^4}
   \mathbb E\!
   \left[
      \phi_{2t}(y-e^{-t}Z_j)
   \right]
   \notag\\
   &=
   B_t(f_j\ast\phi_{2t})(y),
   \label{eq:first-derivative-majorant}
\end{align}
where
\[
   B_t
   =
   \frac{\sup_{y\in\R}U_t(y)}
        {(1-e^{-2t})^4}.
\]

Moreover,
\[
   f_j\ast\phi_{2t}
   \longrightarrow
   f\ast\phi_{2t}
   \quad\text{pointwise and in }L^1.
\]
Lemma~\ref{lem:generalized-dct}, applied to
\eqref{eq:first-derivative-majorant}, yields
\begin{equation}
   \frac{(q_j')^4}{q_j^3}
   \longrightarrow
   \frac{(q')^4}{q^3}
   \qquad\text{in }L^1.
   \label{eq:first-derivative-L1}
\end{equation}

For the second derivative, we have
\begin{equation}
   q_j''(y)
   =
   -\frac{1}{1-e^{-2t}}q_j(y)
   +
   \frac{1}{(1-e^{-2t})^2}
   \mathbb E\!\left[
      (y-e^{-t}Z_j)^2
      \phi_t(y-e^{-t}Z_j)
   \right].
\end{equation}
Therefore,
\begin{align}
   (q_j''(y))^2
   &\le
   \frac{2}{(1-e^{-2t})^2}q_j(y)^2
   \notag\\
   &\quad+
   \frac{2}{(1-e^{-2t})^4}
   \left(
      \mathbb E\!\left[
         (y-e^{-t}Z_j)^2
         \phi_t(y-e^{-t}Z_j)
      \right]
   \right)^2.
\end{align}
By the Cauchy--Schwarz inequality,
\begin{align}
   (q_j''(y))^2
   &\le
   \frac{2}{(1-e^{-2t})^2}q_j(y)^2
   +
   \frac{2q_j(y)}{(1-e^{-2t})^4}
   \mathbb E\!\left[
      (y-e^{-t}Z_j)^4
      \phi_t(y-e^{-t}Z_j)
   \right]
   \notag\\
   &\le
   \frac{2}{(1-e^{-2t})^2}q_j(y)^2
   +
   \frac{2q_j(y)\sup_{y\in\R}U_t(y)}
        {(1-e^{-2t})^4}
   \mathbb E\!\left[
      \phi_{2t}(y-e^{-t}Z_j)
   \right]
   \notag\\
   &=
   \frac{2}{(1-e^{-2t})^2}q_j(y)^2
   +
   2q_j(y)B_t(f_j\ast\phi_{2t})(y).
\end{align}
Thus,
\begin{equation}
   \frac{(q_j'')^2}{q_j}
   \le
   \frac{2}{(1-e^{-2t})^2}q_j
   +
   2B_t(f_j\ast\phi_{2t}).
   \label{eq:second-derivative-majorant}
\end{equation}
The right-hand side converges in $L^1$ to the corresponding expression
with $q$ and $f$. Hence, another application of
Lemma~\ref{lem:generalized-dct} yields
\begin{equation}
   \frac{(q_j'')^2}{q_j}
   \longrightarrow
   \frac{(q'')^2}{q}
   \qquad\text{in }L^1.
   \label{eq:second-derivative-L1}
\end{equation}

Finally,
\begin{equation}
   \left[
      \frac{q_j''}{q_j}
      -
      \left(\frac{q_j'}{q_j}\right)^2
   \right]^2q_j
   \le
   2\frac{(q_j'')^2}{q_j}
   +
   2\frac{(q_j')^4}{q_j^3}.
   \label{eq:K-integrand-majorant}
\end{equation}
The majorant in \eqref{eq:K-integrand-majorant} converges in $L^1$ by
\eqref{eq:first-derivative-L1}--\eqref{eq:second-derivative-L1}.
Combining \eqref{eq:K-integrand-ae} with
Lemma~\ref{lem:generalized-dct} proves that
\[
   K(q_j)\longrightarrow K(q)
\]
along the extracted subsequence. Since every subsequence has a further
subsequence with this limit, the core statement follows for the full
sequence.

We now pass from $\mathcal I_{a,\ell}$ to its closure. Let
$p_n,p\in\overline{\mathcal I}_{a,\ell}$ with $p_n\to p$ in $L^1$.
For each fixed $n$, choose $p_{n,k}\in\mathcal I_{a,\ell}$ such that
\[
   p_{n,k}\longrightarrow p_n
   \qquad\text{in }L^1_{1+x^2}.
\]
The core statement gives
\[
   K(p_{n,k})\longrightarrow K(p_n).
\]
Choose $k(n)$ so that
\[
   \|p_{n,k(n)}-p_n\|_{L^1}<\frac1n,
   \qquad
   |K(p_{n,k(n)})-K(p_n)|<\frac1n.
\]
Set $q_n=p_{n,k(n)}$. Then $q_n\in\mathcal I_{a,\ell}$ and
$q_n\to p$ in $L^1$. Applying the core statement once more gives
$K(q_n)\to K(p)$, and hence
\[
   |K(p_n)-K(p)|
   \le
   |K(p_n)-K(q_n)|
   +
   |K(q_n)-K(p)|
   \longrightarrow0.
\]
\end{proof}

\section{Proof of the conditional multivariate Lindeberg--Feller proposition}
\label{app:conditional-lindeberg-proof}

\begin{proof}[Proof of Proposition~\ref{prop:conditional-clt}]
First, the matrices $C_i$ are i.i.d. and have integrable entries.  Indeed,
\[
   |(C_i)_{jk}|
   \le
   \frac12\E[X_{i,j}^2+X_{i,k}^2\mid\eta_i].
\]
Hence the strong law of large numbers, applied entrywise, gives
\begin{equation}
   V_n=\frac1n\sum_{i=1}^n C_i
   \longrightarrow \E C_1=\Id
   \qquad\text{almost surely}.
   \label{eq:conditional-covariance-convergence}
\end{equation}
Since $d$ is fixed, entrywise convergence is equivalent to convergence in
any matrix norm.

Next, by the tower property and identical distribution,
\begin{align}
   \E\Lambda_n(\varepsilon)
   &=\E\left[
      \|X_1\|^2
      \mathbf1_{\{\|X_1\|>\varepsilon\sqrt n\}}
   \right]
   \longrightarrow0,
   \label{eq:conditional-lindeberg-L1}
\end{align}
where the limit follows from $\E\|X_1\|^2=d<\infty$.  Thus
$\Lambda_n(\varepsilon)\to0$ in $L^1$, and hence in probability, for each
$\varepsilon>0$.

We now explain why the ordinary deterministic multivariate
Lindeberg--Feller theorem can be applied after conditioning.
Let
\[
\mathcal G_n=\sigma(\etaVec_n).
\]
Since
$(\mathbb R^d,\mathcal B(\mathbb R^d))$ is a standard Borel space,
there exists a regular conditional distribution of $X_k$ given
$\mathcal G_n$; see, e.g., \cite[Theorem~8.5]{Kallenberg2021}.
Thus, for each $1\le k\le n$, we may choose a probability kernel
\[
K_{n,k}(\omega,dx)
=
\mathbb P(X_k\in dx\mid\mathcal G_n)(\omega).
\]
Therefore, for every integrable measurable function $\varphi$,
\[
\mathbb E[\varphi(X_k)\mid\mathcal G_n](\omega)
=
\int_{\mathbb R^d}\varphi(x)\,K_{n,k}(\omega,dx)
\]
for almost every $\omega$.

Since the pairs $(X_i,\eta_i)$ are independent, the random vectors
$X_1,\ldots,X_n$ are conditionally independent given $\mathcal G_n$.
Consequently, for almost every $\omega$, a version of their conditional
joint distribution is
\[
K_n(\omega,dx_1,\ldots,dx_n)
=
\bigotimes_{k=1}^n K_{n,k}(\omega,dx_k).
\]
Therefore, once $\omega$ is fixed outside a null set, the conditional
expectations appearing in $V_n$ and $\Lambda_n(\varepsilon)$ become
ordinary deterministic integrals. 
Moreover,
\[
   \E[X_i\mid\mathcal G_n]
   =
   \E[X_i\mid\eta_i]
   =
   0,
   \qquad 1\le i\le n.
\]
Hence, under \(K_n(\omega,\cdot)\), the \(n\)th row
\[
   \left\{\frac{X_i}{\sqrt n}:1\le i\le n\right\}
\]
is an ordinary triangular array of independent centered
\(\R^d\)-valued random vectors.

Fix an arbitrary subsequence $\{n_k\}$. Since
$\Lambda_n(\varepsilon)\to0$ in probability for every
$\varepsilon>0$, for each fixed $\varepsilon>0$ the subsequence
$\{\Lambda_{n_k}(\varepsilon)\}$ admits a further subsequence converging
to zero almost surely. By a diagonal argument over the countable set
$\mathbb Q_{>0}$, we can extract a further subsequence
$\{n_{k_j}\}$ such that
\begin{equation}
    \Lambda_{n_{k_j}}(\varepsilon)\longrightarrow0
    \qquad\text{a.s.}
\end{equation}
for every positive rational $\varepsilon$. Since
$\Lambda_n(\varepsilon)$ is nonincreasing in $\varepsilon$, the same
convergence holds for every $\varepsilon>0$. Moreover,
\eqref{eq:conditional-covariance-convergence} yields
\[
    V_{n_{k_j}}\longrightarrow \Id
    \qquad\text{a.s.}
\]
along this subsequence.

Fix an environment in this probability-one event.  For each \(j\), consider
the deterministic triangular array
\[
   Y_{j,i}
   \triangleq
   \frac{X_i}{\sqrt{n_{k_j}}},
   \qquad
   1\le i\le n_{k_j},
\]
under the conditional product measure
\(K_{n_{k_j}}(\omega,\cdot)\).
The variables in each row are independent and centered.  Moreover, their
total covariance matrix is
\[
   \sum_{i=1}^{n_{k_j}}
   \Cov_{K_{n_{k_j}}(\omega,\cdot)}(Y_{j,i})
   =
   V_{n_{k_j}}(\omega)
   \longrightarrow \Id,
\]
while, for every \(\varepsilon>0\),
\begin{align}
   &\sum_{i=1}^{n_{k_j}}
   \E_{K_{n_{k_j}}(\omega,\cdot)}
   \left[
      \|Y_{j,i}\|^2
      \mathbf 1_{\{\|Y_{j,i}\|>\varepsilon\}}
   \right]
   \notag\\
   &\qquad
   =
   \Lambda_{n_{k_j}}(\varepsilon)(\omega)
   \longrightarrow0.
\end{align}
These are precisely the covariance and Lindeberg conditions in the
multivariate Lindeberg--Feller theorem; see, e.g.,
\cite[Proposition~2.27]{vanDerVaart1998}.
Consequently,
\[
   \frac{1}{\sqrt{n_{k_j}}}
   \sum_{i=1}^{n_{k_j}}X_i
   \Longrightarrow
   N(0,\Id)
\]
under \(K_{n_{k_j}}(\omega,\cdot)\), or equivalently,
\[
   \nu_{n_{k_j}}
   \Longrightarrow
   \gamma_d
   \qquad\text{for almost every }\omega.
\]
Since \(d_{\mathrm{BL}}\) metrizes weak convergence,
\[
   d_{\mathrm{BL}}
   \bigl(\nu_{n_{k_j}},\gamma_d\bigr)
   \longrightarrow0
   \qquad\text{almost surely}.
\]
As every subsequence of
\(\{d_{\mathrm{BL}}(\nu_n,\gamma_d)\}\) admits a further subsequence
converging to zero almost surely, the subsequence criterion for convergence
in probability yields
\[
   d_{\mathrm{BL}}(\nu_n,\gamma_d)
   \overset{P}{\longrightarrow}0.
\]

\end{proof}

\section{Gaussian smoothing and Hessian identities}
\label{app:gaussian-smoothing}

\subsection{Derivative convergence under Gaussian convolution}
\label{app:gaussian-smoothing-proof}

\begin{proof}[Proof of Lemma~\ref{lem:C2conv}]
Fix $t>0$ and a multi-index $\alpha$ with $|\alpha|\le2$.
Since $D^\alpha\phi_t$ is bounded, differentiation under the integral
sign gives, for every $x\in\mathbb R^d$,
\[
D^\alpha(\mu_n*\phi_t)(x)
=
\int_{\mathbb R^d}
D^\alpha\phi_t(x-z)\,\mu_n(dz),
\]
and similarly,
\[
D^\alpha(\mu*\phi_t)(x)
=
\int_{\mathbb R^d}
D^\alpha\phi_t(x-z)\,\mu(dz).
\]

Fix $R>0$. For each $x\in B_R$, define
\[
g_x(z)
\triangleq
D^\alpha\phi_t(x-z),
\qquad z\in\mathbb R^d.
\]
Then $g_x$ is bounded and Lipschitz. Indeed,
\[
\|g_x\|_\infty
=
\sup_{z\in\mathbb R^d}
|D^\alpha\phi_t(x-z)|
=
\sup_{y\in\mathbb R^d}
|D^\alpha\phi_t(y)|
<\infty.
\]
Moreover, by the mean-value theorem,
\[
\operatorname{Lip}(g_x)
\leq
\sup_{z\in\mathbb R^d}
\|\nabla g_x(z)\|
=
\sup_{y\in\mathbb R^d}
\|\nabla D^\alpha\phi_t(y)\|
<\infty.
\]
The last two quantities are finite because every derivative of the
Gaussian density is a polynomial multiplied by a Gaussian density.
Hence there exists a constant $C_{R,d,t,\alpha}<\infty$ such that
\[
\|g_x\|_{\mathrm{BL}}=\|g_x\|_{\infty}+\operatorname{Lip}(g_x)
\leq
C_{R,d,t,\alpha},
\qquad
x\in B_R.
\]

By the definition of the bounded-Lipschitz distance, for every bounded
Lipschitz function $f$,
\[
\left|
\int_{\mathbb R^d} f(z)\,(\mu_n-\mu)(dz)
\right|
\le
\|f\|_{\mathrm{BL}}\,
d_{\mathrm{BL}}(\mu_n,\mu).
\]
Therefore, for every $x\in B_R$,
\begin{align}
&\left|
D^\alpha(\mu_n*\phi_t)(x)
-
D^\alpha(\mu*\phi_t)(x)
\right|
\notag\\
&\quad=
\left|
\int_{\mathbb R^d}
D^\alpha\phi_t(x-z)\,
(\mu_n-\mu)(dz)
\right|
\notag\\
&\quad=
\left|
\int_{\mathbb R^d}
g_x(z)\,
(\mu_n-\mu)(dz)
\right|
\notag\\
&\quad\le
\|g_x\|_{\mathrm{BL}}\,
d_{\mathrm{BL}}(\mu_n,\mu)
\notag\\
&\quad\le
C_{R,d,t,\alpha}\,
d_{\mathrm{BL}}(\mu_n,\mu).
\end{align}
Taking the supremum over $x\in B_R$ yields
\[
\sup_{x\in B_R}
\left|
D^\alpha(\mu_n*\phi_t)(x)
-
D^\alpha(\mu*\phi_t)(x)
\right|
\le
C_{R,d,t,\alpha}\,
d_{\mathrm{BL}}(\mu_n,\mu).
\]
Since
\[
d_{\mathrm{BL}}(\mu_n,\mu)\longrightarrow0,
\]
we obtain
\[
\sup_{x\in B_R}
\left|
D^\alpha(\mu_n*\phi_t)(x)
-
D^\alpha(\mu*\phi_t)(x)
\right|
\longrightarrow0.
\]
Since $R>0$ is arbitrary, the convergence is locally uniform on
$\mathbb R^d$.
\end{proof}

\subsection{Posterior covariance and Hessian bound}
\label{app:hessian-proof}

\begin{proof}[Proof of Lemma~\ref{lem:hessian}]
Let $\mu=\Law(X)$, and let
\[
\phi_t(u)
=
\frac{1}{(2\pi t)^{d/2}}
\exp\left(-\frac{\|u\|^2}{2t}\right),
\qquad u\in\R^d,
\]
denote the density of $N(0,t\Id)$, where $\|\cdot\|$ denotes the Euclidean
norm on $\R^d$. Since
\[
Y=X+\sqrt t\,G,
\]
the density of $Y$ is
\[
p_t(y)
=
\int_{\mathbb R^d}\phi_t(y-x)\,\mu(dx).
\]
Since $\phi_t>0$, we have $p_t(y)>0$ for every $y\in\mathbb R^d$.
Hence the posterior mean can be written as
\begin{equation}
m(y)
=
\E[X\mid Y=y]
=
\frac{
\displaystyle\int_{\mathbb R^d}
x\,\phi_t(y-x)\,\mu(dx)
}{
p_t(y)
}.
\label{eq:posterior-mean-expression}
\end{equation}

Since
\[
\nabla_y\phi_t(y-x)
=
\frac{x-y}{t}\phi_t(y-x),
\]
we obtain
\begin{equation}
\nabla p_t(y) = \frac{p_t(y)}{t}
\bigl(m(y)-y\bigr). 
\end{equation}
Therefore,
\[
\nabla\log p_t(y)
=
\frac{\nabla p_t(y)}{p_t(y)}
=
\frac{m(y)-y}{t},
\]
which proves \eqref{eq:tweedie}.

For the derivative of the posterior mean, write componentwise
\[
   m_k(y)=
   \frac{\int x_k\phi_t(y-x)\,\mu(dx)}{p_t(y)}.
\]
Differentiating the quotient gives
\[
   \frac{\partial m_k(y)}{\partial y_j}
   =\frac1t
   \left(
     \E[X_kX_j\mid Y=y]
     -\E[X_k\mid Y=y]\E[X_j\mid Y=y]
   \right).
\]
Thus
\[
   \nabla m(y)=\frac1t\Cov(X\mid Y=y).
\]
Now differentiating \eqref{eq:tweedie}, we obtain
\begin{align}
\nabla^2\log p_t(y)
&=
\frac1t
\left(
\nabla m(y)-\Id
\right)
\notag\\
&=
-\frac1t\Id
+
\frac1{t^2}
\Cov(X\mid Y=y),
\end{align}
which proves \eqref{eq:hessianposterior}.

Since
\[
\Cov(X\mid Y=y)\succeq0,
\]
equation \eqref{eq:hessianposterior} immediately gives
\[
\nabla^2\log p_t(y)
\succeq
-\frac1t\Id.
\]
If the law of X is log-concave, then its convolution with a Gaussian density is also log-concave \cite{saumard2014log}.
 Hence $p_t$ is log-concave. Since $p_t$ is smooth
and strictly positive,
\[
\nabla^2\log p_t(y)\preceq0.
\]
Therefore,
\[
-\frac1t\Id
\preceq
\nabla^2\log p_t(y)
\preceq
0.
\]
Thus every eigenvalue $\lambda_i(y)$ of
$\nabla^2\log p_t(y)$ satisfies
\[
-\frac1t\le\lambda_i(y)\le0.
\]
Since $\nabla^2\log p_t(y)$ is symmetric, its Hilbert--Schmidt norm
satisfies
\[
   \|\nabla^2\log p_t(y)\|_{\mathrm{HS}}^2
   =
   \sum_{i=1}^d \lambda_i(y)^2
   \le \frac{d}{t^2}.
\]
Integrating against $p_t(y)\,dy$ gives
\[
K(Y)
=
\int_{\mathbb R^d}
\|\nabla^2\log p_t(y)\|_{\HS}^2p_t(y)\,dy
\le
\frac d{t^2},
\]
which proves \eqref{eq:Kbound}.
\end{proof}

\section{Proof of the finite-constellation weak-perturbation lemma}
\label{app:finite-constellation-weak-input}

This appendix proves
Lemma~\ref{lem:conditional-weak-perturbation}.
The proof uses the first-order differentiability of translations in the
Sobolev space naturally associated with the weak Fisher information,
followed by a local KL--Hellinger comparison under a bounded likelihood
ratio.  The latter is a special bounded-likelihood-ratio form of the
standard local equivalence of \(f\)-divergences; see, e.g.,
\cite{sason2018fdivergences}.  We give the argument in full because the
present setting allows nonsmooth and compactly supported densities.

\subsection{Notation}

Set
\begin{equation}
   \mu(dx,dv)
   \triangleq
   dx\,P_V(dv),
   \qquad
   p(x,v)
   \triangleq
   f_{X\mid V}(x\mid v),
   \qquad
   g(x,v)
   \triangleq
   \sqrt{p(x,v)}.
   \label{eq:app-product-space-density}
\end{equation}

Throughout this appendix, the following norm conventions are used.

For a vector \(z\in\mathbb R^k\),
\[
   \|z\|_{\mathbb R^k}
   \triangleq
   \left(
      \sum_{\ell=1}^k z_\ell^2
   \right)^{1/2}.
\]
For a scalar- or complex-valued quantity \(a\), \(|a|\) denotes its
absolute value or modulus.

For a scalar measurable function \(f\) on
\(\mathbb R^d\times\mathcal V\), define
\begin{equation}
   \|f\|_{L^2(\mu)}^2
   \triangleq
   \int_{\mathcal V}
   \int_{\mathbb R^d}
   |f(x,v)|^2
   \,dx\,P_V(dv).
   \label{eq:app-scalar-L2-norm}
\end{equation}
For an \(\mathbb R^M\)-valued measurable function
\(F=(F_1,\ldots,F_M)\), define
\begin{equation}
   \|F\|_{L^2(\mu;\mathbb R^M)}^2
   \triangleq
   \int_{\mathcal V}
   \int_{\mathbb R^d}
   \|F(x,v)\|_{\mathbb R^M}^2
   \,dx\,P_V(dv).
   \label{eq:app-vector-L2-norm}
\end{equation}

We write \(H_x^1(\mu)\) for the class of functions
\(f\in L^2(\mu)\) whose weak derivatives
\(\partial_{x_1}f,\ldots,\partial_{x_d}f\) all belong to
\(L^2(\mu)\).  Thus the Sobolev regularity is imposed only in the
Euclidean coordinate \(x\), and no differentiability in \(v\) is
required.

For densities \(f,q\) with respect to \(\mu\), define
\begin{align}
   D_\mu(f\|q)
   &\triangleq
   \int
   f\log\frac{f}{q}\,d\mu,
   \label{eq:app-KL-definition}\\
   H_\mu^2(f,q)
   &\triangleq
   \int
   \left(
      \sqrt f-\sqrt q
   \right)^2d\mu.
   \label{eq:app-Hellinger-definition}
\end{align}
The convention in \eqref{eq:app-Hellinger-definition} differs by a
constant factor from some definitions of the squared Hellinger distance.

\subsection{Weak Fisher information and the square-root density}

By the weak definition of Fisher information,
\(g(\cdot,v)\in H^1(\mathbb R^d)\) for \(P_V\)-almost every \(v\), and
Fubini's theorem gives
\begin{equation}
   4
   \int_{\mathcal V}
   \int_{\mathbb R^d}
   \|\nabla_x g(x,v)\|_{\mathbb R^d}^2
   \,dx\,P_V(dv)
   =
   \E I(X\mid V)
   <
   \infty.
   \label{eq:app-weak-fisher-sobolev}
\end{equation}
Since
\[
   \|g\|_{L^2(\mu)}^2
   =
   \int p\,d\mu
   =
   1,
\]
we have
\[
   g\in H_x^1(\mu).
\]

Whenever \(\E I(X\mid V)<\infty\), define the conditional score by
\[
   \rho_{X\mid V}(x\mid v)
   \triangleq
   \begin{cases}
      \displaystyle
      \frac{2\nabla_x g(x,v)}{g(x,v)},
      & p(x,v)>0,\\[1.2ex]
      0,
      & p(x,v)=0.
   \end{cases}
\]
By the locality property of weak derivatives,
\[
   \nabla_x g=0
   \qquad
   \mu\text{-a.e. on }\{g=0\}=\{p=0\}.
\]
Therefore,
\begin{equation}
   \nabla_x g(x,v)
   =
   \frac12
   \rho_{X\mid V}(x\mid v)\,
   g(x,v)
   \qquad
   \mu\text{-a.e.}
   \label{eq:app-score-square-root}
\end{equation}
Consequently,
\begin{equation}
   \overline{\Jmat}
   \triangleq
   \E\Jmat(X\mid V)
   =
   \int
   \rho_{X\mid V}(x\mid v)
   \rho_{X\mid V}(x\mid v)^{\mathsf T}
   p(x,v)\,
   \mu(dx,dv).
   \label{eq:app-averaged-fisher-matrix}
\end{equation}

Let
\[
   \P(U=u_j)=\pi_j>0,
   \qquad
   j=1,\ldots,M,
\]
and assume
\begin{equation}
   \sum_{j=1}^M
   \pi_j u_j
   =
   0.
   \label{eq:app-zero-mean-constellation}
\end{equation}

For \(t\ge0\), define
\begin{align}
   p_{j,t}(x,v)
   &\triangleq
   p(x-tu_j,v),
   \label{eq:app-shifted-density}\\
   g_{j,t}(x,v)
   &\triangleq
   \sqrt{p_{j,t}(x,v)}
   =
   g(x-tu_j,v).
   \label{eq:app-shifted-root-density}
\end{align}

\subsection{First-order Sobolev expansion of the translated densities}

For each \(j\), define the scalar function
\begin{equation}
   a_j(x,v)
   \triangleq
   -u_j^{\mathsf T}\nabla_xg(x,v).
   \label{eq:app-aj-definition}
\end{equation}
Since \(u_j\in\mathbb R^d\) is fixed and
\(\nabla_xg\in L^2(\mu;\mathbb R^d)\),
we have \(a_j\in L^2(\mu)\).

We first prove
\begin{equation}
   \frac{g_{j,t}-g}{t}
   \longrightarrow
   a_j
   \qquad
   \text{in }L^2(\mu),
   \qquad
   t\downarrow0.
   \label{eq:app-translation-L2-derivative}
\end{equation}

Let \(\mathcal F_x\) denote the unitary partial
Fourier--Plancherel transform on \(L^2(\mu)\), acting only in the
\(x\)-coordinate. It is normalized so that, for
\(h\in L^1(\mathbb R^d)\cap L^2(\mathbb R^d)\),
\[
   (\mathcal F_xh)(\xi)
   =
   \frac{1}{(2\pi)^{d/2}}
   \int_{\mathbb R^d}
   e^{-i\xi^{\mathsf T}x}
   h(x)\,dx,
\]
and is extended to \(L^2(\mathbb R^d)\) by density. Write
\begin{equation}
   \widehat g
   \triangleq
   \mathcal F_xg.
   \label{eq:app-partial-fourier-transform}
\end{equation}

For \(P_V\)-almost every \(v\), the usual translation and weak
derivative identities hold in \(L^2(\mathbb R^d)\):
\[
   \mathcal F_x\!\left[
      g(\,\cdot-tu_j,v)
   \right](\xi)
   =
   e^{-it\xi^{\mathsf T}u_j}
   \widehat g(\xi,v),
\]
and
\[
   \mathcal F_x\!\left[
      \partial_{x_k}g(\,\cdot,v)
   \right](\xi)
   =
   i\xi_k\widehat g(\xi,v),
   \qquad
   k=1,\ldots,d.
\]
Consequently,
\[
   \mathcal F_x\!\left[
      u_j^{\mathsf T}\nabla_xg(\,\cdot,v)
   \right](\xi)
   =
   i\xi^{\mathsf T}u_j\,
   \widehat g(\xi,v)
\]
in \(L^2(\mathbb R^d)\).

Moreover, the partial Fourier--Plancherel transform is unitary on
\(L^2(\mu)\). Hence, for every \(h\in L^2(\mu)\),
\[
   \int_{\mathcal V}
   \int_{\mathbb R^d}
   |h(x,v)|^2
   \,dx\,P_V(dv)
   =
   \int_{\mathcal V}
   \int_{\mathbb R^d}
   |(\mathcal F_xh)(\xi,v)|^2
   \,d\xi\,P_V(dv).
\]
Applying this identity to the relevant difference quotient gives
\begin{equation}
   \left\|
      \frac{
         g(\,\cdot-tu_j,\cdot)-g
      }{t}
      +
      u_j^{\mathsf T}\nabla_xg
   \right\|_{L^2(\mu)}^2
   =
   \int_{\mathcal V}
   \int_{\mathbb R^d}
   \left|
      \frac{
         e^{-it\xi^{\mathsf T}u_j}-1
      }{t}
      +
      i\xi^{\mathsf T}u_j
   \right|^2
   |\widehat g(\xi,v)|^2
   \,d\xi\,P_V(dv).
   \label{eq:app-fourier-translation}
\end{equation}

For every fixed \(\xi\in\mathbb R^d\),
\[
   \frac{
      e^{-it\xi^{\mathsf T}u_j}-1
   }{t}
   +
   i\xi^{\mathsf T}u_j
   \longrightarrow
   0, \quad \text{as} \quad t \rightarrow 0.
\]
Moreover,
\[
   |e^{-iz}-1|
   \le
   |z|,
   \qquad z\in\mathbb R,
\]
so
\begin{align}
   \left|
      \frac{
         e^{-it\xi^{\mathsf T}u_j}-1
      }{t}
      +
      i\xi^{\mathsf T}u_j
   \right|
   &\le
   2|\xi^{\mathsf T}u_j|
   \notag\\
   &\le
   2
   \|\xi\|_{\mathbb R^d}
   \|u_j\|_{\mathbb R^d}.
   \label{eq:app-fourier-dominating-bound}
\end{align}
Thus the integrand in
\eqref{eq:app-fourier-translation} is bounded by
\[
   4
   \|u_j\|_{\mathbb R^d}^2
   \|\xi\|_{\mathbb R^d}^2
   |\widehat g(\xi,v)|^2.
\]
The Fourier characterization of \(H_x^1(\mu)\) gives
\begin{equation}
   \int_{\mathcal V}
   \int_{\mathbb R^d}
   \|\xi\|_{\mathbb R^d}^2
   |\widehat g(\xi,v)|^2
   \,d\xi\,P_V(dv)
   <
   \infty.
   \label{eq:app-fourier-H1-integrability}
\end{equation}
Therefore dominated convergence proves
\eqref{eq:app-translation-L2-derivative}.

Equivalently,
\begin{equation}
   g_{j,t}
   =
   g+t a_j+r_{j,t},
   \qquad
   \|r_{j,t}\|_{L^2(\mu)}
   =
   o(t).
   \label{eq:app-shifted-root-expansion}
\end{equation}

By the locality property of weak gradients,
\[
   \nabla_xg=0
   \qquad
  \mu\text{-a.e. on }\{g=0\}.
\]
Since \(\{g=0\}=\{p=0\}\), it follows that
\begin{equation}
   a_j=0
   \qquad
   \mu\text{-a.e. on }\{p=0\}.
   \label{eq:app-aj-zero-on-zero-set}
\end{equation}

\subsection{First-order cancellation in the output mixture}

The density of \((X+tU,V)\) with respect to \(\mu\) is
\begin{equation}
   \bar p_t(x,v)
   \triangleq
   \sum_{j=1}^M
   \pi_jp_{j,t}(x,v),
   \label{eq:app-mixture-density}
\end{equation}
and define
\begin{equation}
   \bar g_t(x,v)
   \triangleq
   \sqrt{\bar p_t(x,v)}.
   \label{eq:app-mixture-root-density}
\end{equation}

We claim that
\begin{equation}
   \frac{\bar g_t-g}{t}
   \longrightarrow
   0
   \qquad
   \text{in }L^2(\mu).
   \label{eq:app-mixture-zero-first-derivative}
\end{equation}

Define the \(\mathbb R^M\)-valued functions
\begin{align}
   G_t(x,v)
   &\triangleq
   \bigl(
      \sqrt{\pi_1}\,g_{1,t}(x,v),
      \ldots,
      \sqrt{\pi_M}\,g_{M,t}(x,v)
   \bigr),
   \label{eq:app-Gt-definition}\\
   A(x,v)
   &\triangleq
   \bigl(
      \sqrt{\pi_1}\,a_1(x,v),
      \ldots,
      \sqrt{\pi_M}\,a_M(x,v)
   \bigr).
   \label{eq:app-A-definition}
\end{align}
Also set
\begin{equation}
   e
   \triangleq
   (\sqrt{\pi_1},\ldots,\sqrt{\pi_M})
   \in\mathbb R^M.
   \label{eq:app-e-definition}
\end{equation}
Since \(\sum_j\pi_j=1\),
\[
   \|e\|_{\mathbb R^M}=1.
\]
At \(t=0\),
\[
   G_0(x,v)=g(x,v)e.
\]

Moreover, for every fixed \((x,v)\),
\begin{align}
   \|G_t(x,v)\|_{\mathbb R^M}
   &=
   \left(
      \sum_{j=1}^M
      \pi_jg_{j,t}(x,v)^2
   \right)^{1/2}
   \notag\\
   &=
   \sqrt{\bar p_t(x,v)}
   =
   \bar g_t(x,v).
   \label{eq:app-mixture-as-euclidean-norm}
\end{align}
Thus
\eqref{eq:app-mixture-as-euclidean-norm}
is a pointwise identity involving the Euclidean norm in
\(\mathbb R^M\).

Since the constellation contains only finitely many points,
\eqref{eq:app-translation-L2-derivative} implies
\begin{equation}
   \frac{G_t-G_0}{t}
   \longrightarrow
   A
   \qquad
   \text{in }L^2(\mu;\mathbb R^M).
   \label{eq:app-vector-L2-derivative}
\end{equation}
Indeed,
\begin{equation}
   \left\|
      \frac{G_t-G_0}{t}-A
   \right\|_{L^2(\mu;\mathbb R^M)}^2=
   \sum_{j=1}^M
   \pi_j
   \left\|
      \frac{g_{j,t}-g}{t}-a_j
   \right\|_{L^2(\mu)}^2
   \longrightarrow0.
\end{equation}

We next linearize the pointwise Euclidean norm.
For almost every \((x,v)\) such that \(g(x,v)>0\),
the map
\[
   z\longmapsto\|z\|_{\mathbb R^M}
\]
is differentiable at \(z=g(x,v)e\neq0\), with gradient \(e\).
Hence
\begin{equation}
   \frac{
      \|g(x,v)e+tA(x,v)\|_{\mathbb R^M}
      -
      g(x,v)
   }{t}
   \longrightarrow
   e^{\mathsf T}A(x,v).
   \label{eq:app-pointwise-norm-derivative}
\end{equation}
On the set \(\{g=0\}\), relation
\eqref{eq:app-aj-zero-on-zero-set} implies
\(A=0\) almost everywhere, so
\eqref{eq:app-pointwise-norm-derivative}
also holds there.

By the reverse triangle inequality in \(\mathbb R^M\),
\begin{align}
   \left|
      \frac{
         \|ge+tA\|_{\mathbb R^M}
         -
         g
      }{t}
   \right|
   &=
   \frac{
      \left|
         \|ge+tA\|_{\mathbb R^M}
         -
         \|ge\|_{\mathbb R^M}
      \right|
   }{t}
   \notag\\
   &\le
   \|A\|_{\mathbb R^M}.
   \label{eq:app-pointwise-norm-dominating-bound}
\end{align}
Since
\(A\in L^2(\mu;\mathbb R^M)\),
the square of the right-hand side is integrable.
Therefore dominated convergence yields
\begin{equation}
   \frac{
      \|G_0+tA\|_{\mathbb R^M}
      -
      g
   }{t}
   \longrightarrow
   e^{\mathsf T}A
   \qquad
   \text{in }L^2(\mu).
   \label{eq:app-linearized-mixture-norm}
\end{equation}

The Euclidean norm is \(1\)-Lipschitz:
\[
   \left|
      \|z\|_{\mathbb R^M}
      -
      \|w\|_{\mathbb R^M}
   \right|
   \le
   \|z-w\|_{\mathbb R^M},
   \qquad
   z,w\in\mathbb R^M.
\]
Hence
\begin{align}
   \left\|
      \frac{
         \|G_t\|_{\mathbb R^M}
         -
         \|G_0+tA\|_{\mathbb R^M}
      }{t}
   \right\|_{L^2(\mu)}\le
   \left\|
      \frac{G_t-G_0}{t}-A
   \right\|_{L^2(\mu;\mathbb R^M)}
   \longrightarrow0.
   \label{eq:app-norm-lipschitz}
\end{align}

Finally,
\begin{align}
   e^{\mathsf T}A
   &=
   \sum_{j=1}^M
   \pi_j a_j
   \notag\\
   &=
   -
   \left(
      \sum_{j=1}^M
      \pi_j u_j
   \right)^{\mathsf T}
   \nabla_xg
   \notag\\
   &=
   0
   \label{eq:app-mixture-cancellation}
\end{align}
by \eqref{eq:app-zero-mean-constellation}.
Combining
\eqref{eq:app-mixture-as-euclidean-norm},
\eqref{eq:app-linearized-mixture-norm},
\eqref{eq:app-norm-lipschitz}, and
\eqref{eq:app-mixture-cancellation}
proves
\eqref{eq:app-mixture-zero-first-derivative}.

Consequently, by \eqref{eq:app-translation-L2-derivative} and \eqref{eq:app-mixture-zero-first-derivative},   for every fixed \(j\),
\begin{equation}
   \frac{
      g_{j,t}-\bar g_t
   }{t}
   \longrightarrow
   a_j
   \qquad
   \text{in }L^2(\mu).
   \label{eq:app-component-mixture-root-difference}
\end{equation}

By definition of the squared Hellinger distance,
\begin{align}
   H_\mu^2(p_{j,t},\bar p_t)
   &=
   \|g_{j,t}-\bar g_t\|_{L^2(\mu)}^2
   \notag\\
   &=
   t^2
   \|a_j\|_{L^2(\mu)}^2
   +
   o(t^2).
   \label{eq:app-hellinger-expansion}
\end{align}

Using
\eqref{eq:app-score-square-root} and
\eqref{eq:app-averaged-fisher-matrix},
\begin{align}
   \|a_j\|_{L^2(\mu)}^2
   &=
   \int
   \left|
      u_j^{\mathsf T}\nabla_xg(x,v)
   \right|^2
   \,\mu(dx,dv)
   \notag\\
   &=
   \frac14
   \int
   \left|
      u_j^{\mathsf T}
      \rho_{X\mid V}(x\mid v)
   \right|^2
   p(x,v)
   \,\mu(dx,dv)
   \notag\\
   &=
   \frac14
   u_j^{\mathsf T}
   \overline{\Jmat}
   u_j.
   \label{eq:app-aj-fisher}
\end{align}
Therefore
\begin{equation}
   H_\mu^2(p_{j,t},\bar p_t)
   =
   \frac{t^2}{4}
   u_j^{\mathsf T}
   \overline{\Jmat}
   u_j
   +
   o(t^2).
   \label{eq:app-hellinger-fisher-expansion}
\end{equation}

\subsection{A local KL--Hellinger comparison}

We next prove the local comparison needed to pass from
\eqref{eq:app-hellinger-fisher-expansion}
to relative entropy.

\begin{lemma}[Local KL--Hellinger comparison under a bounded likelihood ratio]
\label{lem:app-local-kl-hellinger}
Let \(f_t,q_t,p\) be probability densities with respect to a common measure
\(\mu\), where \(t\downarrow0\).  Assume
\begin{align}
   \|\sqrt{f_t}-\sqrt p\|_{L^2(\mu)}
   &\longrightarrow0,
   \label{eq:app-ft-hellinger-p}\\
   \|\sqrt{q_t}-\sqrt p\|_{L^2(\mu)}
   &\longrightarrow0.
   \label{eq:app-qt-hellinger-p}
\end{align}
Assume also that there exists \(C<\infty\) such that
\begin{equation}
   f_t
   \le
   Cq_t
   \qquad
   \mu\text{-a.e.}
   \label{eq:app-likelihood-ratio-bound}
\end{equation}
for all sufficiently small \(t\).

Suppose further that
\begin{equation}
   d_t
   \triangleq
   \frac{
      \sqrt{f_t}-\sqrt{q_t}
   }{t}
   \longrightarrow
   d
   \qquad
   \text{in }L^2(\mu),
   \label{eq:app-dt-L2}
\end{equation}
where
\begin{equation}
   d=0
   \qquad
   \mu\text{-a.e. on }\{p=0\}.
   \label{eq:app-d-zero}
\end{equation}
Then
\begin{equation}
   D_\mu(f_t\|q_t)
   =
   2H_\mu^2(f_t,q_t)
   +
   o(t^2).
   \label{eq:app-KL-Hellinger-local-equivalence}
\end{equation}
\end{lemma}

\begin{proof}
On \(\{q_t>0\}\), define
\[
   r_t
   \triangleq
   \frac{f_t}{q_t}.
\]
On \(\{q_t=0\}\), set \(r_t=0\).
By \eqref{eq:app-likelihood-ratio-bound},
\(f_t=0\) on \(\{q_t=0\}\), so this convention does not change any
integral.  Moreover,
\begin{equation}
   0\le r_t\le C
   \qquad
   \mu\text{-a.e.}
   \label{eq:app-ratio-uniform-bound}
\end{equation}

Since \(f_t\) and \(q_t\) are probability densities,
\[
   \int q_t(r_t-1)\,d\mu
   =
   \int(f_t-q_t)\,d\mu
   =
   0.
\]
Hence
\begin{equation}
   D_\mu(f_t\|q_t)
   =
   \int
   q_t\Phi(r_t)\,d\mu,
   \label{eq:app-KL-Phi}
\end{equation}
where
\[
   \Phi(r)
   \triangleq
   r\log r-r+1,
   \qquad r\ge0,
\]
with the convention \(0\log0=0\).

Likewise,
\begin{equation}
   H_\mu^2(f_t,q_t)
   =
   \int
   q_t\Psi(r_t)\,d\mu,
   \label{eq:app-H-Psi}
\end{equation}
where
\[
   \Psi(r)
   \triangleq
   (\sqrt r-1)^2.
\]

For \(r\ne1\), define
\[
   c(r)
   \triangleq
   \frac{\Phi(r)}{\Psi(r)}.
\]
Since
\[
   \Phi(1+s)
   =
   \frac{s^2}{2}
   +
   o(s^2),
   \qquad
   \Psi(1+s)
   =
   \frac{s^2}{4}
   +
   o(s^2),
   \qquad
   s\to0,
\]
we have
\[
   \lim_{r\to1}c(r)=2.
\]
Set
\[
   c(1)=2.
\]
Furthermore,
\[
   \Phi(0)=1,
   \qquad
   \Psi(0)=1,
\]
so set
\[
   c(0)=1.
\]
Thus \(c\) is continuous on \([0,C]\).  Define
\begin{equation}
   C_0
   \triangleq
   \sup_{0\le r\le C}
   |c(r)-2|
   <
   \infty.
   \label{eq:app-c-uniform-bound}
\end{equation}

Since
\[
   q_t\Psi(r_t)
   =
   \left(
      \sqrt{f_t}-\sqrt{q_t}
   \right)^2
   =
   t^2d_t^2,
\]
equations
\eqref{eq:app-KL-Phi} and
\eqref{eq:app-H-Psi} give
\begin{equation}
   \frac{
      D_\mu(f_t\|q_t)
      -
      2H_\mu^2(f_t,q_t)
   }{t^2}
   =
   \int
   \bigl(c(r_t)-2\bigr)
   d_t^2\,d\mu.
   \label{eq:app-KL-H-difference}
\end{equation}

From
\eqref{eq:app-dt-L2},
\begin{align}
   \|d_t^2-d^2\|_{L^1(\mu)}
   &\le
   \|d_t-d\|_{L^2(\mu)}
   \left(
      \|d_t\|_{L^2(\mu)}
      +
      \|d\|_{L^2(\mu)}
   \right)
   \notag\\
   &\longrightarrow0.
   \label{eq:app-dt-square-L1}
\end{align}

Noth that 
\begin{equation}
   \left|
      \int
      \bigl(c(r_t)-2\bigr)d_t^2\,d\mu
   \right|\le
   C_0
   \|d_t^2-d^2\|_{L^1(\mu)}
   +
   \int
   |c(r_t)-2|\,d^2\,d\mu.
   \label{eq:app-final-KL-H-bound}
\end{equation}
It remains to prove
\begin{equation}
    \int
   |c(r_t)-2|\,d^2\,d\mu
   \longrightarrow0.  
   \label{eq:app-cr}
\end{equation}
We use a subsequence argument. Let \(\{t_n\}\) be an arbitrary sequence satisfying
\(t_n\downarrow0\).  By
\eqref{eq:app-ft-hellinger-p} and
\eqref{eq:app-qt-hellinger-p}, after passing to a subsequence, still denoted
by \(\{t_n\}\), we may assume that
\[
   \sqrt{f_{t_n}}
   \longrightarrow
   \sqrt p,
   \qquad
   \sqrt{q_{t_n}}
   \longrightarrow
   \sqrt p
   \qquad
   \mu\text{-a.e.}
\]
Indeed, each \(L^2(\mu)\)-convergent sequence admits an
almost-everywhere convergent subsequence, and a successive extraction
gives the two convergences simultaneously.

On the set \(\{p>0\}\), the limiting denominator is strictly positive.
Hence, for \(\mu\)-almost every point in \(\{p>0\}\),
\[
   r_{t_n}
   =
   \frac{f_{t_n}}{q_{t_n}}
   =
   \left(
      \frac{\sqrt{f_{t_n}}}{\sqrt{q_{t_n}}}
   \right)^2
   \longrightarrow1.
\]
By continuity of \(c\) at \(1\),
\[
   c(r_{t_n})-2
   \longrightarrow0
   \qquad
   \mu\text{-a.e. on }\{p>0\}.
\]
On the other hand,
\eqref{eq:app-d-zero} gives
\[
   d=0
   \qquad
   \mu\text{-a.e. on }\{p=0\}.
\]
Therefore
\[
   |c(r_{t_n})-2|\,d^2
   \longrightarrow0
   \qquad
   \mu\text{-a.e.}
\]

By \eqref{eq:app-ratio-uniform-bound} and
\eqref{eq:app-c-uniform-bound},
\[
   |c(r_{t_n})-2|\,d^2
   \le
   C_0d^2.
\]
Since \(d\in L^2(\mu)\), we have \(d^2\in L^1(\mu)\).
The dominated convergence theorem therefore yields
\begin{equation*}
      \int
   |c(r_{t_n})-2|\,d^2\,d\mu
   \longrightarrow0
\end{equation*}
along the extracted subsequence.

For any sequence \(\{t_n\}\) with \(t_n\downarrow0\), there exists a
subsequence \(\{t_{n_k}\}\) such that
\[
   \int
   |c(r_{t_{n_k}})-2|\,d^2\,d\mu
   \longrightarrow0
   \qquad
   \text{as }k\to\infty.
\]
This implies convergence along the full parameter \(t\downarrow0\).
Indeed, otherwise there would exist \(\varepsilon>0\) and a sequence
\(s_n\downarrow0\) such that
\[
   \int |c(r_{s_n})-2|\,d^2\,d\mu
   \ge\varepsilon
   \qquad\text{for every }n.
\]
No subsequence of \(\{s_n\}\) could then have the integral converging
to zero, contradicting the preceding subsequence property. Hence
\begin{equation}
      \label{eq:app-c-times-d2}
         \int |c(r_t)-2|\,d^2\,d\mu
   \longrightarrow0,
   \qquad t\downarrow0.
\end{equation}

Finally, by \eqref{eq:app-dt-square-L1} and 
\eqref{eq:app-c-times-d2}, we obtain
\begin{equation}
\lim_{t\rightarrow 0} 
   \left|
      \int
      \bigl(c(r_t)-2\bigr)d_t^2\,d\mu
   \right| =0
\end{equation}
Hence the right-hand side of
\eqref{eq:app-KL-H-difference}
converges to zero.  This proves
\eqref{eq:app-KL-Hellinger-local-equivalence}.
\end{proof}

\subsection{Completion of the proof of
Lemma~\ref{lem:conditional-weak-perturbation}}

Fix \(j\in\{1,\ldots,M\}\).
Since
\[
   \bar p_t
   =
   \sum_{k=1}^M
   \pi_kp_{k,t},
\]
we have
\begin{equation}
   \bar p_t
   \ge
   \pi_jp_{j,t},
   \qquad
   \frac{p_{j,t}}{\bar p_t}
   \le
   \frac1{\pi_j}
   \quad
   \mu\text{-a.e.}
   \label{eq:app-mixture-likelihood-bound}
\end{equation}

Furthermore,
\eqref{eq:app-translation-L2-derivative} gives
\begin{equation}
   \|\sqrt{p_{j,t}}-\sqrt p\|_{L^2(\mu)}
   =
   O(t)
   \longrightarrow0,
   \label{eq:app-component-to-base-Hellinger}
\end{equation}
while
\eqref{eq:app-mixture-zero-first-derivative} gives
\begin{equation}
   \|\sqrt{\bar p_t}-\sqrt p\|_{L^2(\mu)}
   =
   o(t)
   \longrightarrow0.
   \label{eq:app-mixture-to-base-Hellinger}
\end{equation}

Moreover,
\eqref{eq:app-component-mixture-root-difference} gives
\[
   \frac{
      \sqrt{p_{j,t}}-\sqrt{\bar p_t}
   }{t}
   \longrightarrow
   a_j
   \qquad
   \text{in }L^2(\mu),
\]
and
\eqref{eq:app-aj-zero-on-zero-set} gives
\[
   a_j=0
   \qquad
   \mu\text{-a.e. on }\{p=0\}.
\]

Thus
Lemma~\ref{lem:app-local-kl-hellinger}
applies with
\[
   f_t=p_{j,t},
   \qquad
   q_t=\bar p_t,
   \qquad
   d=a_j.
\]
Together with
\eqref{eq:app-hellinger-fisher-expansion}, it yields
\begin{align}
   D_\mu(p_{j,t}\|\bar p_t)
   &=
   2H_\mu^2(p_{j,t},\bar p_t)
   +
   o(t^2)
   \notag\\
   &=
   \frac{t^2}{2}
   u_j^{\mathsf T}
   \overline{\Jmat}
   u_j
   +
   o(t^2).
   \label{eq:app-component-KL-expansion}
\end{align}

Conditionally on \(U=u_j\), the pair
\((X+tU,V)\) has density \(p_{j,t}\) with respect to \(\mu\),
whereas its unconditional density is \(\bar p_t\).
Therefore
\begin{align}
   I(U;X+tU,V)
   &=
   \sum_{j=1}^M
   \pi_j
   D_\mu(p_{j,t}\|\bar p_t)
   \notag\\
   &=
   \frac{t^2}{2}
   \sum_{j=1}^M
   \pi_j
   u_j^{\mathsf T}
   \overline{\Jmat}
   u_j
   +
   o(t^2).
   \label{eq:app-mutual-information-sum}
\end{align}
The final remainder is \(o(t^2)\) because \(M<\infty\).

Since
\[
   Q
   =
   \Cov(U)
   =
   \sum_{j=1}^M
   \pi_j
   u_ju_j^{\mathsf T},
\]
we have
\begin{align}
   \sum_{j=1}^M
   \pi_j
   u_j^{\mathsf T}
   \overline{\Jmat}
   u_j
   &=
   \sum_{j=1}^M
   \pi_j
   \tr\!\left(
      \overline{\Jmat}
      u_ju_j^{\mathsf T}
   \right)
   \notag\\
   &=
   \tr\!\left(
      \overline{\Jmat}Q
   \right)
   \notag\\
   &=
   \tr\!\left(
      Q\overline{\Jmat}
   \right).
   \label{eq:app-trace-identity}
\end{align}
Therefore
\begin{equation}
   I(U;X+tU,V)
   =
   \frac{t^2}{2}
   \tr\!\left(
      Q\,\E\Jmat(X\mid V)
   \right)
   +
   o(t^2).
   \label{eq:app-final-t-expansion}
\end{equation}

Finally, because \(U\) is independent of \(V\),
\[
   I(U;X+tU,V)
   =
   I(U;V)
   +
   I(U;X+tU\mid V)
   =
   I(U;X+tU\mid V).
\]
Setting \(t=\sqrt\rho\) in
\eqref{eq:app-final-t-expansion}
gives
\[
   I(U;X+\sqrt\rho\,U\mid V)
   =
   \frac{\rho}{2}
   \tr\!\left(
      Q\,\E\Jmat(X\mid V)
   \right)
   +
   o(\rho),
\]
which proves
Lemma~\ref{lem:conditional-weak-perturbation}.
\qed

\end{document}